\documentclass[11pt]{article}%
\usepackage{amsfonts}
\usepackage{amssymb}
\usepackage{amsmath}
\usepackage{geometry}
\usepackage{graphicx}
\usepackage{setspace}
\usepackage{booktabs}
\usepackage{bbm}%
\providecommand{\U}[1]{\protect\rule{.1in}{.1in}}
\newtheorem{theorem}{Theorem}

\newtheorem{lemma}[theorem]{Lemma}

\newenvironment{proof}[1][Proof]{\noindent\textbf{#1.} }{\ \rule{0.5em}{0.5em}}
\newcommand{\ul}[1]{\underline{#1}}
\newcommand{\EE}{\mathbb{E}}

\begin{document}

\title{ }

\begin{center}
\textit{Technical Report}

{\LARGE Marginal Structural Models }

{\LARGE for Time-varying Endogenous Treatments:\ }

{\LARGE A Time-Varying Instrumental Variable Approach}

{\Large Eric J. Tchetgen Tchetgen}

{\Large Haben Michael }

{\Large Yifan Cui } $\ $

{\large Department of Statistics}

{\large The Wharton School\ }

{\large University of Pennsylvania}

\textbf{Abstract}
\end{center}

\noindent Robins (1998) introduced marginal structural models (MSMs),\ a
general class of counterfactual models for the joint effects of time-varying
treatment regimes in complex longitudinal studies subject to time-varying
confounding. He established identification of MSM parameters under a
sequential randomization assumption (SRA), which essentially rules out
unmeasured confounding of treatment assignment over time. \ In this technical
report, we consider sufficient conditions for identification of MSM parameters
with the aid of a time-varying instrumental variable, when sequential
randomization fails to hold due to unmeasured confounding. Our identification
conditions essentially require that no unobserved confounder predicts
compliance type for the time-varying treatment, the longitudinal
generalization of the identifying condition of Wang and Tchetgen Tchetgen
(2018). Under this assumption, We derive a large class of semiparametric
estimators that extends standard inverse-probability weighting (IPW), the most
popular approach for estimating MSMs under SRA, by incorporating the
time-varying IV through a modified set of weights. The set of influence
functions for MSM parameters is derived under a semiparametric model with sole
restriction on observed data distribution given by the MSM, and is shown to
provide a rich class of multiply robust estimators, including a local
semiparametric efficient estimator.

%----------------------------------------------------------------------

\bigskip\pagebreak

\section{\noindent Introduction}

Robins (1998,1999, 2000a) introduced a new class of counterfactual models
known as marginal structural models (MSMs) that encode the joint causal
effects of time-varying treatment subject to time-varying confounding.\ For
identification, Robins relied on a sequential randomization assumption (SRA)
which essentially rules out unmeasured confounding of the time-varying
treatment. In this technical report, we consider sufficient conditions for
identification of MSM parameters with the aid of a time-varying instrumental
variable, when sequential randomization fails to hold due to unmeasured
confounding. Our identification conditions essentially require longitudinal
generalizations of (i) IV relevance, (ii) exclusion restriction, and (iii)
IV\ independence assumptions, together with a key assumption (iv) that no
unobserved confounder predicts compliance type for the time-varying treatment,
a longitudinal generalization of the identification condition of Wang and
Tchetgen Tchetgen (2018). Under these assumptions, we derive a large class of
semiparametric estimators which extends standard inverse-probability weighting
(IPW), the most common approach for estimating MSMs under SRA (Robins et al.
2000, Hern\'{a}n et al, 2000, 2001), that incorporates the time-varying IV
through a modified set of weights. The set of influence functions for MSM
parameters under IV identification is derived for a semiparametric model with
sole restriction on the observed data distribution given by the MSM, and is
shown to provide a rich class of multiply robust estimators, including a
locally semiparametric efficient estimator.

Prior to the current work, Robins (1994) developed a general framework for
identification and estimation of causal effects of time-varying endogenous
treatments using a time-varying instrumental variable under a structural
nested model (SNM).\ As described in Robins (2000a), parameters of an SNM can
under certain conditions be interpreted as MSM parameters, in which case,
Robins (1994) provides alternative identification conditions to ours.\ In
contrast, the proposed methodology is more general as it directly targets MSM
parameters irrespective of whether or not they can be interpreted as
parameters of an equivalent SNM. \ 

\section{Notation and definitions}

Continuous time is denoted by $t$ and is measured in months since the
beginning of a subject's follow-up. The index $j$ is often used when we wish
to indicate an integer number of months. $J$ corresponds to the administrative
end of follow-up, recorded in whole months. \ Notice that as staggered entry
of participants is a common feature of longitudinal studies, $J$ is considered
random. We use capital letters to represent random variables and lower-case
letters to represent possible realizations (values) of random variables.
\ $A(j)$ denotes a binary treatment taken by a subject in ($j,j+1]$, and
$L(j)$ is a vector of relevant prognostic factors for outcomes
$Y(j+1),...Y(J)$. We assume that recorded data on the treatment and prognostic
factors do not change except at these times; moreover, $L(j)$ temporally
precedes $A(j)$, and $Y(j)$ is included in $L(j).$ \ For any time dependent
variable, we use overbars to denote the history of that variable up to and
including $t;$ for example, the covariate process through $t$ is $\overline
{L}(t)=\left\{  L(0),L(1),...,L(t)\right\}  =\left\{
L(0),L(1),...,L(int\left[  t\right]  )\right\}  $ where $int\left[  t\right]
$ is the greatest integer less than or equal to $t.$ \ Note that throughout,
unless necessary, we suppress the subscript denoting individual, because we
assume that the random vector for each subject is drawn independently from a
distribution common to all subjects. \ We use the symbol $\amalg$ to indicate
statistical independence; for example $A\amalg B|D$ means that $A$ is
conditionally independent of $B$ given $D.$ \ Finally, for any $O_{i}$, define
$\mathbb{P}_{n}\left[  O\right]  =\sum\limits_{i=1}^{n}O_{i}/n$.

In order to formally define MSMs, we need to introduce counterfactual or
potential outcomes. \ Neyman (1923) was the first to use counterfactual
outcomes to analyze the causal effect of time independent treatments in
randomized experiments. \ Later on, Rubin (1974)\ and Holland (1986) adopted
Neyman's idea and demonstrated the usefulness of counterfactuals in the
analysis of the causal effects of time-independent treatments from
observational data. \ Robins (1986,1987)\ proposed a formal counterfactual
theory of causal inference that extended Neyman's time-independent treatment
theory to longitudinal studies with both direct and indirect effects and
sequential time-varying treatments and confounders. \ Throughout, we assume no
censoring, although we note that methods to address dependent censoring
described in Robins (1998) can easily be adapted to our setting. \ For a
specific fixed treatment history $\overline{a}=\left(  a\left(  0\right)
,a\left(  1\right)  ,...,a\left(  J-1\right)  \right)  ,\overline
{L}_{\overline{a}}$ is defined to be the random vector representing a
subject's covariate process had (possibly contrary to fact) the subject been
treated [i.e through time $J-1]$ with the particular treatment regime
$\overline{a}$ rather than his or her observed treatment history $\overline
{A}=\overline{A}\left(  J-1\right)  .$ Note that $\overline{a}\left(
t\right)  $ is a possible realization of the random variable $\overline
{A}\left(  t\right)  $. \ For each possible history $\overline{a},$ we are
assuming that a subject's potential covariate/outcome process $\left\{
\overline{L}_{\overline{a}}\right\}  $ is well defined, although generally
unobserved. \ Each individual therefore has a corresponding set of
counterfactual variables $\mathcal{L}_{\mathcal{A}}=\left\{  \overline
{L}_{\overline{a}}=\overline{L}_{\overline{a}}\left(  J\right)  :\overline
{a}\in\mathcal{A}\right\}  $ where $\mathcal{A}$ is the support of
$\overline{A},$ and throughout, in accordance with reality, the future cannot
cause the past, i.e., $L_{\overline{a}}(j)=L_{\overline{a}(j-1)}(j),
j=1,\ldots,J-1.$

\section{Brief Review of MSM inference under sequential randomization}

\ An MSM for $\left\{  \overline{Y}_{\overline{a}}:\overline{a}\in
\mathcal{A}\right\}  $ places restriction on the marginal distribution of the
$\overline{Y}_{\overline{a}}$ possibly conditional on baseline variables $V\in
L(0)$.\ Robins (1998, 2000) describes a large number of MSMs reproduced below;
however we note that this is certainly not exhaustive:

\textbf{Model 1. in Models 1.1-1.3, }suppose that $Y(j)\equiv0$ for $j<J$ and
$Y\equiv Y\left(  J)\right)  $ at the end of follow-up $J=K+1$ w.p.1. for a
constant $K.$

\textbf{Model 1.1: Non-linear least-squares: }$E\left(  Y_{\overline{a}%
}|V\right)  =g\left(  \overline{a},V;\beta_{0}\right)  ,$ where $g\left(
\cdot,\cdot;\cdot\right)  $ is a known function.

\textbf{Model 1.2: Semiparametric Regression: }$\eta\left(  E\left(
Y_{\overline{a}}|V\right)  \right)  =g\left(  \overline{a},V;\beta_{0}\right)
+g^{\ast}\left(  V\right)  ,$ where $\eta$ is a known monotone link function,
$g^{\ast}$ is an unknown unrestricted function, and $g\left(  \cdot
,\cdot;\cdot\right)  $ is a known function with $g\left(  0,\cdot
;\cdot\right)  =0$.

\textbf{Model 1.3. Stratified Transformation model:}$\Pr\left(  R\left(
\overline{a},V;\beta_{0}\right)  \leq r|V\right)  =F_{0}\left(  r|V\right)  ,$
$F_{0}$ is an unknown distribution function, $R\left(  \overline{a}%
,V;\beta_{0}\right)  =r\left(  Y_{\overline{a}},\overline{a},V;\beta\right)  $
is a known increasing function of $Y_{\overline{a}}$ satisfying $r\left(
y,\overline{a},V;\beta\right)  =y$ if $\overline{a}=0$ or $\beta=0.$

\textbf{Model 1.4. Multivariate non-linear least squares: }$\ $Suppose that
the outcome is observed longitudinally, so that the MSM restricts the marginal
joint distribution%
\[
\left\{  \overline{Y_{\overline{a}}}\left(  K+1\right)  :\overline{a}\right\}
=\{Y_{a(0)}(1),Y_{\overline{a}(1)}(2),...,Y_{\overline{a}(K)}(K+1):\overline
{a}\};
\]
the multivariate non-linear least squares MSM specifies
\[
E\left(  Y_{\overline{a}}(m)|V\right)  =g_{m}\left(  \overline{a}%
(m-1),V;\beta_{0}\right)  ,m=1,...,K+1,
\]
where $g_{m}$ are known functions.

\textbf{Model 2. }Suppose that $J=\infty\,,$ and $Y_{\overline{a}}$ is a
failure time process which jumps from $\ 0$ to $1$ at some particular time and
stays at $1$ thereafter. Define the failure time $T_{\overline{a}}$ by the
equation $Y_{\overline{a}}\left(  T_{\overline{a}}\right)  =1$ and
$Y_{\overline{a}}\left(  T_{\overline{a}}^{-}\right)  =0$. Let $\lambda
_{W}\left(  t\right)  $ denote the hazard function of $W.$

\textbf{Model 2.1.Cox Proportional Hazards model }%
\[
\lambda_{T_{\overline{a}}}(t|V)=\lambda_{0}\left(  t\right)  \exp\left(
r\left(  \overline{a}\left(  t^{-}\right)  ,t,\beta_{0},V\right)  \right)  ,
\]
where $r()$ is a known function which satisfies $r\left(  \overline
{\mathbf{0}},t,\beta,0\right)  =0.$

\textbf{Model 2.2.Stratified Cox Proportional Hazards model }%
\[
\lambda_{T_{\overline{a}}}(t|V)=\lambda_{0}\left(  t|V\right)  \exp\left(
r\left(  \overline{a}\left(  t^{-}\right)  ,t,\beta_{0},V\right)  \right)  ,
\]
where $r()$ is a known function which satisfies $r\left(  \overline
{\mathbf{0}},t,\beta,V\right)  =0.$

\textbf{Model 2.3.Stratified time-dependent Accelerated Failure Time model.}%
\[
\Pr\left(  R\left(  \overline{a},V;\beta_{0}\right)  \leq r|V\right)
=F_{0}\left(  r|V\right)  ,
\]
$F_{0}$ is an unknown distribution function, $R\left(  \overline{a}%
,V;\beta_{0}\right)  =r\left(  T_{\overline{a}},\overline{a},V;\beta\right)  $
is a known increasing function of $Y_{\overline{a}}$ satisfying $r\left(
y,\overline{a},V;\beta\right)  =y$ if $\overline{a}=0.$

Other MSMs possibly of interest include quantile MSMs, additive hazards MSMs
and restricted residual mean survival MSMs. \ While these and other possible
MSMs are not discussed herein, our results readily extend to these MSMs.
Having defined the underlying set of counterfactual variables and MSMs of
interest, we now consider how they relate to the observed data. \ Three
important assumptions are essential to the identification of the\ MSM
parameter $\beta_{0}$ from the observed data.\ \ First, of the many
counterfactual variables in $\mathcal{Y}_{\mathcal{A}}$, only one is
ultimately observed in a given individual. \ In fact, we observe a realization
of $\overline{Y}_{\overline{a}}$ only if the treatment history $\overline{a}$
is equal to a subject's actual treatment history $\overline{A};$ that is
$\overline{Y}=\overline{Y}_{\overline{A}}$ w.p.1. \ This identity constitutes
the fundamental \textquotedblright consistency\textquotedblright\ assumption
that links the counterfactual data $\overline{Y}_{\overline{a}}$ to the
observed data $(\overline{Y},\overline{A}).$ \ The next assumption is that
there are no unmeasured confounders for the effect of $A(j)$ on $\overline{Y}%
$, that is, for all treatment histories $\overline{a},$%

\begin{equation}
\overline{Y}_{\overline{a}}\amalg A(j)|\overline{A}(j-1)=\overline
{a}(j-1),\overline{L}(j),\text{ \ }j=1,\ldots,J. \label{sequential rand}%
\end{equation}
\ \ This assumption generalizes Rosembaum and Rubin's (1983) assumption of
ignorable treatment assignment to longitudinal studies with time-varying
treatments and confounders and is also referred to as the sequential
randomization assumption (SRA) (Robins, 1998). \ It states that, conditional
on treatment history and the history of all recorded covariates up to $j$,
treatment at $j$ is independent of the counterfactual random variables
$Y_{\overline{a}}\left(  j+1\right)  ,\ldots,Y_{\overline{a}}\left(  J\right)
.$ This will be true if, for example, all prognostic factors for $\overline
{Y}$ used by the physicians to determine whether treatment $A$ is given at $j$
are recorded in $\left(  \overline{A}(j-1),\overline{L}(j)\right)  .$ \ For
example, physicians generally check HIV infected patients' current CD4 count
before deciding whether or not he or she needs to initiate HAART (highly
active antiretroviral therapy) to delay death or progression to AIDS.
\ Clearly, because CD4 count also correlates with the patient's time of death
or progression to AIDS, the assumption of no unmeasured confounders would be
false if $\overline{L}(j)$ did not include patients' current CD4 count. \ 

In an observational study, the assumption of no unmeasured confounder cannot
be guaranteed to hold, and it is not subject to empirical test. However it
will hold to a reasonable approximation if good efforts are made to collect
data on the crucial covariates. \ Investigating the sensitivity to violations
of SRA through a formal sensitivity analysis is important but will not be
discussed in this paper. \ Robins, Greenland, and Hu (1999), and Robins,
Rotnitzky and Scharfstein (2000), have provided details on the theory of
sensitivity analysis in causal models. Below, we will consider instrumental
variable methods when SRA fails to hold.

We finally assume that the following positivity assumption holds. \ For all
$a(j)$ in the support of $A(j)$%
\[
\text{if }f\left(  \overline{L}(j),\overline{A}(j-1)\right)  >0\text{ then
}f(a(j)|\overline{L}(j),\overline{A}(j-1))>0.
\]
This assumption essentially states that if any set of subjects at time $j$
have the opportunity of continuing on a treatment regime $\overline{a}$ under
consideration, at least some will take that opportunity. Positivity is
actually a sufficient but not a necessary condition to apply the methods
described in this paper; see ref. (Robins 1998) for further details.

\ Consider the semiparametric model $\mathcal{M}_{tp}$ where (i) the treatment
process
\begin{equation}
f\left(  A(k)=1|\overline{L}(k),\overline{A}(k-1)\right)  ,\text{
}k=0,....,J-1,\text{is known,} \label{treatment}%
\end{equation}
with (ii) observed data $O=\left(  \overline{A}=\overline{A}\left(
J-1\right)  ,\overline{L}=\overline{L}\left(  J\right)  \right)  ;$ and (iii)
an MSM with target parameter $\beta_{0}$. Also define $\mathcal{M}_{tp}^{\ast
}$ as $\mathcal{M}_{tp}$ where in the data generating mechanism (ii),
(\ref{treatment}) is replaced with user-specified density (\ref{treatment}%
$^{\ast}$) $f^{\ast}\left(  A(k)=1|V,\overline{A}(k-1)\right)  ,$ and the
model is otherwise identical$.$ As noted by Robins, under $\mathcal{M}%
_{tp}^{\ast},$ MSMs 1-2 simplify to well-known statistical models, where
\textquotedblleft$^{\ast}$\textquotedblright\ denotes expectation under the model.

\textbf{Model 1.1: }$E^{\ast}\left(  Y|\overline{a},V\right)  =g\left(
\overline{a},V;\beta_{0}\right)  ,$ where $g\left(  \cdot,\cdot;\cdot\right)
$ is a known function.

\textbf{Model 1.2: }$\eta\left(  E^{\ast}\left(  Y|\overline{a},V\right)
\right)  =g\left(  \overline{a},V;\beta_{0}\right)  +g^{\ast}\left(  V\right)
$.

\textbf{Model 1.3. }$\Pr^{\ast}\left(  R\left(  \overline{a},V;\beta
_{0}\right)  \leq r|V\right)  =F_{0}^{\ast}\left(  r|V\right)  $, $R\left(
\overline{a},V;\beta_{0}\right)  =r\left(  Y,\overline{a},V;\beta\right)  $.

\textbf{Model 1.4. }$E^{\ast}\left(  Y(m)|\overline{a},V\right)  =g_{m}\left(
\overline{a}(m-1),V;\beta_{0}\right)  ,m=1,...,K+1.$

\textbf{Model 2.1.}$\lambda_{T}^{\ast}(t|\overline{a}\left(  t\right)
,V)=\lambda_{0}\left(  t\right)  \exp\left(  r\left(  \overline{a}\left(
t^{-}\right)  ,t,\beta_{0},V\right)  \right)  .$

\textbf{Model 2.2.}$\lambda_{T}^{\ast}(t|\overline{a}\left(  t\right)
,V)=\lambda_{0}\left(  t|V\right)  \exp\left(  r\left(  \overline{a}\left(
t^{-}\right)  ,t,\beta_{0},V\right)  \right)  $

\textbf{Model 2.3.}$\Pr^{\ast}\left(  R\left(  \overline{a},V;\beta
_{0}\right)  \leq r|V\right)  =F_{0}^{\ast}\left(  r|V\right)  ,$ $R\left(
\overline{a},V;\beta_{0}\right)  =r\left(  T,\overline{a},V;\beta\right)  $

Then, Robins established that all regular and asymptotically linear (RAL)
estimators $\widehat{\beta}\left(  h,\phi\right)  $ in $\mathcal{M}_{tp}$ can
be obtained by solving:%
\[
o_{p}\left(  n^{-1/2}\right)  =\mathbb{P}_{n}\widehat{D}\left(  O;h,\phi
,\beta\right)
\]
with $\mathbb{P}_{n}\widehat{D}\left(  h,\phi,\beta\right)  =\mathbb{P}%
_{n}\widehat{D}_{sm}\left(  h,\beta\right)  /\overline{\mathcal{W}}%
\mathcal{+}D_{tp}\left(  \phi\right)  $,
\begin{align*}
\overline{\mathcal{W}}  &  =%
%TCIMACRO{\dprod \limits_{k=0}^{J-1}}%
%BeginExpansion
{\displaystyle\prod\limits_{k=0}^{J-1}}
%EndExpansion
\mathcal{W}_{k}=%
%TCIMACRO{\dprod \limits_{k=0}^{J-1}}%
%BeginExpansion
{\displaystyle\prod\limits_{k=0}^{J-1}}
%EndExpansion
\frac{f\left(  A(k)|\overline{L}(k),\overline{A}(k-1)\right)  }{f^{\ast
}\left(  A(k)|V,\overline{A}(k-1)\right)  },\\
D_{tp}\left(  \phi\right)   &  =%
%TCIMACRO{\dsum \limits_{k=0}^{J-1}}%
%BeginExpansion
{\displaystyle\sum\limits_{k=0}^{J-1}}
%EndExpansion
\phi\left(  k,\overline{A}\left(  k\right)  ,\overline{L}(k)\right)
-\mathbb{E}\left(  \phi\left(  k,\overline{A}\left(  k\right)  ,\overline
{L}(k)\right)  |\overline{A}\left(  k-1\right)  ,\overline{L}(k)\right)
\end{align*}
and $\mathbb{P}_{n}\widehat{D}_{sm}\left(  h,\beta\right)  =\mathbb{P}%
_{n}V_{sm}^{\ast}\left(  h,\beta\right)  +o_{p}\left(  1\right)  ,$ where
$\left\{  \widehat{D}_{sm}\left(  h,\beta\right)  :h\right\}  $ and $\left\{
V_{sm}^{\ast}\left(  h,\beta\right)  :h\right\}  $ are the following familiar
estimating functions of $\beta$ of models 1-2 under (\ref{treatment}*) and
their associated influence functions$.$

\textbf{Model 1.1: }$\widehat{D}_{sm}\left(  h,\beta\right)  =V_{sm}^{\ast
}\left(  h,\beta\right)  $ where $V_{sm}^{\ast}\left(  h,\beta\right)
=h\left(  \overline{A},V\right)  \varepsilon\left(  \beta\right)
;\varepsilon\left(  \beta\right)  =Y-g\left(  \overline{A},V;\beta_{0}\right)
,$

\textbf{Model 1.2: For }$\eta\left(  x\right)  =x,$ $\widehat{D}_{sm}\left(
h,\beta\right)  =V_{sm}^{\ast}\left(  h,\beta\right)  =\left(  \varepsilon
\left(  \beta\right)  -h_{1}\left(  \overline{A},V\right)  \right)  \left(
h_{2}\left(  \overline{A},V\right)  -\mathbb{E}^{\ast}\left\{  h_{2}\left(
\overline{A},V\right)  |V\right\}  \right)  $. for any choice of $h_{1}$ and
$h_{2}$ of same dimension as $\beta.$ For $\eta(x)=\log(x/\left(  1-x)\right)
,$ let $p\left(  \beta\right)  =$expit$\left(  g\left(  \overline{a}%
,V;\beta_{0}\right)  +g^{\ast}\left(  V\right)  \right)  ,$ $\widehat{p}%
\left(  \beta\right)  =$expit$\left(  g\left(  \overline{a},V;\beta
_{0}\right)  +\widehat{g}^{\ast}\left(  V\right)  \right)  $ where
$\widehat{g}^{\ast}\left(  V\right)  $ is a $n^{1/4}-$ consistent estimatof of
$g^{\ast}\left(  V\right)  .\widehat{D}_{sm}\left(  h,\beta\right)
=\widehat{\varepsilon}\left(  \beta\right)  \left(  h_{2}\left(  \overline
{A},V\right)  -\mathbb{E}^{\ast}\left\{  h_{2}\left(  \overline{A},V\right)
\widehat{p}\left(  \beta\right)  (1-\widehat{p}\left(  \beta\right)
)|V\right\}  /\mathbb{E}^{\ast}\left\{  \widehat{p}\left(  \beta\right)
(1-\widehat{p}\left(  \beta\right)  )|V\right\}  \right)
,\widehat{\varepsilon}\left(  \beta\right)  =Y-\widehat{p}\left(
\beta\right)  ;$

$V_{sm}^{\ast}\left(  h,\beta\right)  =\varepsilon\left(  \beta\right)
\left(  h_{2}\left(  \overline{A},V\right)  -\mathbb{E}^{\ast}\left\{
h_{2}\left(  \overline{A},V\right)  p\left(  \beta\right)  (1-p\left(
\beta\right)  )|V\right\}  /\mathbb{E}^{\ast}\left\{  p\left(  \beta\right)
(1-p\left(  \beta\right)  )|V\right\}  \right)  $

\textbf{Model 1.3.}$\widehat{D}_{sm}\left(  h,\beta\right)  =V_{sm}^{\ast
}\left(  h,\beta\right)  =h\left(  R\left(  \beta_{0}\right)  ,\overline
{A},V\right)  -\int h\left(  R\left(  \beta_{0}\right)  ,\overline
{a},V\right)  dF^{\ast}\left(  \overline{a}|V\right)  $ where [??]

\textbf{Model 1.4. }$\widehat{D}_{sm}\left(  h,\beta\right)  =V_{sm}^{\ast
}\left(  h,\beta\right)  =h\left(  \overline{A},V\right)  \mathbf{\varepsilon
}\left(  \beta\right)  $ where $\mathbf{\varepsilon}\left(  \beta\right)
=\left(  \varepsilon_{1}\left(  \beta\right)  ,....,\varepsilon_{K+1}\left(
\beta\right)  \right)  ,$ $\varepsilon_{m}\left(  \beta\right)  =Y(m)-g_{m}%
\left(  \overline{a}(m-1),V;\beta_{0}\right)  ,m=1,...,K+1,$ and $h\left(
\overline{A},V\right)  $ is of dimension $\dim\left(  \beta\right)
\times\left(  K+1\right)  .$

\textbf{Model 2.1.}%
\[
\widehat{D}_{sm}\left(  h,\beta\right)  =\int dN(t)\left\{  h\left(
t,\overline{A},V\right)  -\frac{\mathbb{P}_{n}^{\ast}\left[  h\left(
t,\overline{A},V\right)  \exp\left(  r\left(  \overline{A}\left(  t\right)
,t,\beta_{0},V\right)  \right)  I(T\geq t)\right]  }{\mathbb{P}_{n}^{\ast
}\left[  \exp\left(  r\left(  \overline{A}\left(  t\right)  ,t,\beta
_{0},V\right)  \right)  I(T\geq t)\right]  }\right\}
\]
$.$%

\[
U_{sm}\left(  h,\beta\right)  =\int dM_{T}(t)\left\{  h\left(  t,\overline
{A},V\right)  -\frac{\mathbb{E}^{\ast}\left[  h\left(  t,\overline
{A},V\right)  \exp\left(  r\left(  \overline{A}\left(  t\right)  ,t,\beta
_{0},V\right)  \right)  I(T\geq t)\right]  }{\mathbb{E}^{\ast}\left[
\exp\left(  r\left(  \overline{A}\left(  t\right)  ,t,\beta_{0},V\right)
\right)  I(T\geq t)\right]  }\right\}  ,
\]
where $N_{T}(t)=I(T\leq t)$ and $dM_{T}(t)=dN(t)-\lambda_{T}(t|\overline
{A},V)I(T\geq t)dt.\,$

\textbf{Model 2.2.} $\widehat{D}_{sm}\left(  h,\beta\right)  $ and
$V_{sm}^{\ast}\left(  h,\beta\right)  $ are as above with $\mathbb{P}%
_{n}^{\ast}\left[  h\left(  t,\overline{A},V\right)  \exp\left(  r\left(
\overline{A}\left(  t\right)  ,t,\beta_{0},V\right)  \right)  I(T\geq
t)|\right]  $ and $\mathbb{P}_{n}^{\ast}\left[  \exp\left(  r\left(
\overline{A}\left(  t\right)  ,t,\beta_{0},V\right)  \right)  I(T\geq
t)\right]  $ in $\widehat{D}_{sm}\left(  h,\beta\right)  $ replaced by an
$n^{-1/4}$ consistent estimator of $\mathbb{\ \ }$ $\mathbb{E}^{\ast}\left[
h\left(  t,\overline{A},V\right)  \exp\left(  r\left(  \overline{A}\left(
t\right)  ,t,\beta_{0},V\right)  \right)  I(T\geq t)|V\right]  $ and
$\mathbb{E}^{\ast}\left[  \exp\left(  r\left(  \overline{A}\left(  t\right)
,t,\beta_{0},V\right)  \right)  I(T\geq t)|V\right]  $.

\textbf{Model 2.3.}%
\begin{align*}
\widehat{D}_{sm}\left(  h,\beta\right)   &  =\int_{0}^{\infty}dtI\left(
R\left(  \beta\right)  \leq t\right)  \left\{  H_{2}\left(  t,\beta\right)
-\mathbb{E}^{\ast}\left[  H_{2}\left(  t,\beta\right)  |V\right]  \right\} \\
&  +\int_{0}^{\infty}dN_{R\left(  \beta\right)  }\left(  t\right)  \left\{
H_{1}\left(  t,\beta\right)  -\mathbb{E}^{\ast}\left[  H_{1}\left(
t,\beta\right)  |V\right]  \right\}
\end{align*}

and for $j=1,2,H_{j}\left(  t,\beta\right)  =h_{j}\left(  t,\overline
{A}\left(  r^{-1}\left(  t,\overline{A},V,\beta\right)  \right)  ,V\right)
;V_{sm}^{\ast}\left(  h,\beta\right)  =D_{sm}\left(  h,\beta\right)
-\mathbb{E}^{\ast}\left[  D_{sm}\left(  h,\beta\right)  |V\right]  .$

Note that in Model 2.1,
\begin{align*}
&  \mathbb{E}^{\ast}\left[  h\left(  t,\overline{A},V\right)  \exp\left(
r\left(  \overline{A}\left(  t\right)  ,t,\beta,V\right)  \right)  I(T\geq
t)\right] \\
&  =\mathbb{E}\left[  \exp\left(  r\left(  \overline{A}\left(  t\right)
,t,\beta,V\right)  \right)  I(T\geq t)/\overline{\mathcal{W}}\left(
int\left(  t\right)  \right)  \}\right]
\end{align*}
and
\begin{align*}
&  \mathbb{E}^{\ast}\left[  \exp\left(  r\left(  \overline{A}\left(  t\right)
,t,\beta,V\right)  \right)  I(T\geq t)\right] \\
&  =\mathbb{E}\left[  \exp\left(  r\left(  \overline{A}\left(  t\right)
,t,\beta,V\right)  \right)  I(T\geq t)\overline{/\mathcal{W}}\left(
int\left(  t\right)  \right)  \right]
\end{align*}
and likewise in Model 2.2
\begin{align*}
&  \mathbb{E}^{\ast}\left[  h\left(  t,\overline{A},V\right)  \exp\left(
r\left(  \overline{A}\left(  t\right)  ,t,\beta,V\right)  \right)  I(T\geq
t)|V\right] \\
&  =\mathbb{E}\left[  \exp\left(  r\left(  \overline{A}\left(  t\right)
,t,\beta,V\right)  \right)  I(T\geq t)/\overline{\mathcal{W}}\left(
int\left(  t\right)  \right)  \}|V\right]
\end{align*}
and%

\begin{align*}
&  \mathbb{E}^{\ast}\left[  \exp\left(  r\left(  \overline{A}\left(  t\right)
,t,\beta,V\right)  \right)  I(T\geq t)|V\right] \\
&  =\mathbb{E}\left[  \exp\left(  r\left(  \overline{A}\left(  t\right)
,t,\beta,V\right)  \right)  I(T\geq t)/\overline{\mathcal{W}}\left(
int\left(  t\right)  \right)  |V\right]  .
\end{align*}
Robins (1998) also established that for fixed $h,$ the optimal choice of
$\phi$ in model $\mathcal{M}_{tp}$ is given by $\phi_{opt}\left(
k,\overline{A}\left(  k\right)  ,\overline{L}(k)\right)  $ $=-\mathbb{E}%
\left(  D_{sm}\left(  h,\beta\right)  /\overline{\mathcal{W}}|\overline
{A}\left(  k\right)  ,\overline{L}(k)\right)  $, in the sense that given $h,$
there is no estimator with asymptotic variance smaller than $\widehat{\beta
}\left(  h,\phi_{opt}\right)  .$ Let $C_{J}=\mathbb{E}\left(  D_{sm}\left(
h,\beta\right)  |\overline{A}\left(  J-1\right)  ,\overline{L}(J-1)\right)
,\ $and define $C_{j}$ recursively as $C_{j}=\sum_{a(j)}\mathbb{E}\left(
C_{j+1}\left(  a_{j}\right)  |\overline{A}\left(  j-1\right)  ,\overline
{L}(j-1)\right)  $ $j=J-1,...,1$. \ As in practice neither $\overline
{\mathcal{W}}$ nor $\left\{  C_{j}:j\right\}  $ are known and must be
estimated using working models which are sufficiently parsimonious to resolve
the curse of dimensionality, e.g. parametric working models, Robins (2000b)
established that in models 1.1-1.4, $\widehat{\beta}_{dr}$ is a doubly robust
(dr) estimator in the sense that it is consistent and asymptotically normal if
either $\widehat{\overline{\mathcal{W}}}$ or $\left\{  \widehat{C}%
_{j}:j\right\}  $ is consistent but not necessarily both, where
$\widehat{\beta}_{dr}$ solves:
\[
0=\mathbb{P}_{n}\widehat{D}\left(  h,\widehat{\phi}_{opt},\beta\right)
=\mathbb{P}_{n}\widehat{D}_{sm}\left(  h,\beta\right)  /\widehat{\overline
{\mathcal{W}}}\mathcal{+}D_{tp}\left(  \widehat{\phi}_{opt}\right)  ,
\]
with
\[
D_{tp}\left(  \widehat{\phi}_{opt}\right)  =-%
%TCIMACRO{\dsum \limits_{j=0}^{J-1}}%
%BeginExpansion
{\displaystyle\sum\limits_{j=0}^{J-1}}
%EndExpansion
\widehat{C}_{j+1}/\widehat{\overline{\mathcal{W}}}\mathcal{(}j)+\sum
_{a(j)}\widehat{C}_{j+1}/\widehat{\overline{\mathcal{W}}}\mathcal{(}j-1),
\]
and
\[
\widehat{\overline{\mathcal{W}}}\mathcal{(}j)=%
%TCIMACRO{\dprod \limits_{k=0}^{j}}%
%BeginExpansion
{\displaystyle\prod\limits_{k=0}^{j}}
%EndExpansion
\widehat{\mathcal{W}}_{k}%
\]
Note that it is likewise possible to construct dr estimators in models
2.1-2.4, however, in Cox MSM 2.1, this requires construction of dr estimators
of \
\begin{align*}
&  \mathbb{E}^{\ast}\left[  h\left(  t,\overline{A},V\right)  \exp\left(
r\left(  \overline{A}\left(  t\right)  ,t,\beta,V\right)  \right)  I(T\geq
t)\right] \\
&  =\mathbb{E}\left[  \exp\left(  r\left(  \overline{A}\left(  t\right)
,t,\beta,V\right)  \right)  I(T\geq t)/\overline{\mathcal{W}}\left(
int\left(  t\right)  \right)  \}\right]  ,
\end{align*}
and
\[
\mathbb{E}^{\ast}\left[  \exp\left(  r\left(  \overline{A}\left(  t\right)
,t,\beta,V\right)  \right)  I(T\geq t)\right]  =\mathbb{E}\left[  \exp\left(
r\left(  \overline{A}\left(  t\right)  ,t,\beta,V\right)  \right)  I(T\geq
t)\overline{\mathcal{W}}\left(  int\left(  t\right)  \right)  \right]  ,
\]
likewise for Model 2.2 which requires a dr estimator of versions of above
quantities conditional on $V$ under $\mathbb{E}^{\ast}$; details are omitted,
however see Tchetgen Tchetgen and Robins (2012) for an illustration in the
case of point exposure. A similar approach applies to model 2.2.

\section{MSM inference with time-varying instrumental variable.}

\subsection{New inverse-probability-of-instrumental-variable weighted
estimators}

\ In this section, we do not make the assumption of sequential randomization
$\left(  \ref{sequential rand}\right)  $ and allow for unmeasured time-varying
covariates $\overline{U}=\left(  U\left(  0\right)  ,....U(J-1)\right)  ,$
such that $U(j)$ is a common cause of $\underline{A}\left(  j\right)  =\left(
A\left(  j\right)  ,...,A\left(  J-1\right)  \right)  $ and $\underline{Y}%
\left(  j+1\right)  =\left(  Y\left(  j+1\right)  ,...,Y\left(  J\right)
\right)  .$ We assume that in addition to $\left(  \overline{L},\overline
{A}\right)  ,$ a binary time-varying instrumental variable $Z(j)$ is observed
just prior to $A(j),j=0,...,J-1;$ further, we assume that had $\overline{U}$
been observed, sequential ignorability would hold. Specifically, we make the
following assumption of latent sequential randomization:
\begin{equation}
\overline{L}_{\overline{a}}\amalg A(j)|\overline{A}(j-1)=\overline
{a}(j-1),\overline{L}(j),\overline{U}(j),\overline{Z}(j)\text{ \ }%
j=1,\ldots,J-1. \label{latent ignorability}%
\end{equation}

However, noting that $\overline{L}_{\overline{a}}\not \amalg A(j)|\overline
{A}(j-1)=\overline{a}(j-1),\overline{L}(j),\overline{Z}(j)$, and given that
$\overline{U}$ is unobserved, the MSM\ is not identified without an additional
assumption. For the purpose of identification, we suppose that $\overline{Z}$
satisfies the following key time-varying IV conditions:

\underline{Assumption (1): IV Relevance}:%
\begin{equation}
Z(j)\not \amalg A(j)|\overline{A}(j-1),\overline{L}(j),\overline{Z}\left(
j-1\right)  \text{ \ }j=1,\ldots,J-1 \label{IV Relevance}%
\end{equation}

\underline{Assumption (2): Exclusion Restriction}:
\begin{equation}
\left(  \overline{L}_{\overline{a}\overline{z}},\overline{U}_{\overline
{a}\overline{z}}\right)  =\left(  \overline{L}_{\overline{a}},\overline
{U}_{\overline{a}}\right)  \text{ a.s.} \label{Exclusion Restriction}%
\end{equation}

\underline{Assumption (3): IV independence }:%

\begin{equation}
\left(  \overline{U}_{\overline{a}},\overline{L}_{\overline{a}}\right)  \amalg
Z(j)|\overline{A}(j-1)=\overline{a}(j-1),\overline{L}(j),\overline
{Z}(j-1)\text{ \ }j=0,\ldots,J-1 \label{IV Independence}%
\end{equation}

\underline{Assumption (4): IV positivity:}%
\[
0<\Pr\left(  Z(j)=1|\overline{A}(j-1),\overline{L}(j),\overline{Z}\left(
j-1\right)  \right)  <1\text{ for \ }j=0,\ldots,J-1
\]

In addition, we suppose the following holds.

\underline{Assumption (5) Independent Compliance Type:}
\begin{align}
&  \mathbb{E}\left[  A(j)|\overline{U}(j),\overline{L}(j),\overline{A}\left(
j-1\right)  ,\overline{Z}(j-1),Z(j)=1\right]  -\mathbb{E}\left[
A(j)|\overline{U}(j),\overline{L}(j),\overline{A}\left(  j-1\right)
,\overline{Z}(j-1),Z(j)=0\right] \label{Independent Compliance type}\\
&  =\delta_{j}\left(  \overline{L}(j),\overline{A}\left(  j-1\right)
,\overline{Z}\left(  j-1\right)  \right) \nonumber
\end{align}
The assumption states that while $\overline{U}(j)$ may confound the causal
effects of $\overline{A}(j),$ no component of $\overline{U}(j)$ interacts with
$Z(j)$ in its additive effects on $A(j).$ A causal interpretation of the
assumption is available if $Z(j)$ $\amalg A_{z(j)}\left(  j\right)
|\overline{U}(j),\overline{L}(j),\overline{A}\left(  j-1\right)  ,\overline
{Z}(j-1)$ in which case $\left(  \ref{Independent Compliance type}\right)  $
implies:
\begin{align}
&  \mathbb{E}\left[  A_{z(j)=1}(j)-A_{z(j)=0}(j)|\overline{U}(j),\overline
{L}(j),\overline{A}\left(  j-1\right)  ,\overline{Z}(j-1)\right] \\
&  =\delta_{j}\left(  \overline{L}(j),\overline{A}\left(  j-1\right)
,\overline{Z}\left(  j-1\right)  \right)  \text{, }j=0,...,J-1.\nonumber
\end{align}
that $\overline{U}(j)$ is conditionally independent of compliance type at time
$j,$ expressed in terms of a person's potential treatment variables under
hypothetical IV interventions $\left\{  A_{z(j)=1}(j),A_{z(j)=0}(j)\right\}
$. This assumption is a longitudinal generalization of a similar assumption
made by Wang and Tchetgen Tchetgen (2018a) and Wang et al (2018b) in the case
of point exposure and IV. \ Below, we will make use of the fact that under our
assumptions, $\left\{  \delta_{j}:j\right\}  $ is empirically identified. Specifically,

\begin{lemma}
Under assumptions (3) and (5), we have that
\[
\delta_{j}(\overline{l}(j),\overline{a}(j-1),\overline{z}(j-1))=\mathbb{E}%
\left[  A(j)|,\overline{l}(j),\overline{a}\left(  j-1\right)  ,\overline
{z}(j-1),Z(j)=1\right]  -\mathbb{E}\left[  A(j)|\overline{l}(j),\overline
{a}\left(  j-1\right)  ,\overline{z}(j-1),Z(j)=0\right]
\]

\end{lemma}

\begin{proof}%
\begin{align*}
P(A(j) &  =1|\overline{l}(j),\overline{a}(j-1),\overline{z}%
(j-1),Z(j)=1)-P(A(j)=1|\overline{l}(j),\overline{a}(j-1),\overline
{z}(j-1),Z(j)=0)\\
&  =\int P(A(j)=1|\overline{l}(j),\overline{a}(j-1),\overline{z}%
(j-1),Z(j)=1,\overline{u}(j))dF(\overline{u}(j)|\overline{l}(j),\overline
{a}(j-1),\overline{z}(j-1),Z(j)=1)\\
&  -\int P\left(  A(j)=1|\overline{l}(j),\overline{a}(j-1),\overline
{z}(j-1),Z(j)=0,\overline{u}(j)\right)  dF\left(  \overline{u}(j)|\overline
{l}(j),\overline{a}(j-1),\overline{z}(j-1),Z(j)=0\right)  \\
&  =\int P(A(j)=1|\overline{l}(j),\overline{a}(j-1),\overline{z}%
(j-1),Z(j)=1,\overline{u}(j))dF(\overline{u}(j)|\overline{l}(j),\overline
{a}(j-1),\overline{z}(j-1))\\
&  -\int P\left(  A(j)=1|\overline{l}(j),\overline{a}(j-1),\overline
{z}(j-1),Z(j)=0,\overline{u}(j)\right)  dF\left(  \overline{u}(j)|\overline
{l}(j),\overline{a}(j-1),\overline{z}(j-1)\right)  \\
&  =\int\delta_{j}(\overline{l}(j),\overline{a}(j-1),\overline{z}%
(j-1))dF\left(  \overline{u}(j)|\overline{l}(j),\overline{a}(j-1),\overline
{z}(j-1)\right)  \\
&  =\delta_{j}(\overline{l}(j),\overline{a}(j-1),\overline{z}(j-1)).
\end{align*}

\end{proof}

We define the following modified time varying weights:%
\begin{align*}
\overline{\mathcal{W}}^{\dag}\left(  j\right)   &  =%
%TCIMACRO{\dprod \limits_{k=1}^{j}}%
%BeginExpansion
{\displaystyle\prod\limits_{k=1}^{j}}
%EndExpansion
\mathcal{W}_{k,1}^{\dag}\mathcal{W}_{k,2}^{\dag}\\
\overline{\mathcal{W}}^{\dag} &  =\overline{\mathcal{W}}^{\dag}\left(
J-1\right)
\end{align*}
where
\[
\mathcal{W}_{k,1}^{\dag}=\frac{f\left(  Z(k)|\overline{L}(k),\overline
{A}(k-1),\overline{Z}(k-1)\right)  \delta_{k}\left(  \overline{L}%
(k),\overline{A}\left(  k-1\right)  ,\overline{Z}\left(  k-1\right)  \right)
}{\left(  -1\right)  ^{1-Z(k)}}%
\]
and
\[
\mathcal{W}_{k,2}^{\dag}=\frac{1}{\left(  -1\right)  ^{1-A(k)}f^{\ast}\left(
A(k)|V,\overline{A}(k-1)\right)  }%
\]

We give our main result.

\begin{lemma}
Suppose that together with consistency, Assumptions (1)-(5) hold. For any
measurable function $G=g(\overline{A},\overline{L}),$
\begin{align*}
\mathbb{E}\left(  g(\overline{A},\overline{L})/\overline{\mathcal{W}}^{\dag
}|V\right)   &  =\sum_{\overline{a}}\mathbb{E}\left\{  g(\overline
{a},\overline{L}_{\overline{a}})|V\right\}
%TCIMACRO{\dprod \limits_{j=0}^{J-1}}%
%BeginExpansion
{\displaystyle\prod\limits_{j=0}^{J-1}}
%EndExpansion
f^{\ast}\left(  a(j)|V,\overline{a}(j-1)\right) \\
&  =\mathbb{E}^{\ast}\left\{  G|V\right\}
\end{align*}

\end{lemma}

Note that the above Lemma continues to hold under the less stringent latent
SRA $Y_{\overline{a}}\amalg A(j)|\overline{A}(j-1)=\overline{a}(j-1),\overline
{L}(j),\overline{U}\left(  j\right)  $ \ $j=1,\ldots,J-1,$ if $g(\overline
{A},\overline{L})$ only depends on $\overline{L}$ through $Y=Y(J).$ The Lemma
motivates the following simple weighted estimating equation of $\beta_{0}$ in
models 1 and 2.\ Suppose that one has obtained $n^{1/2}$-consistent estimators
$\widehat{f}\left(  Z(k)|\overline{L}(k),\overline{A}(k-1),\overline
{Z}(k-1)\right)  $ and $\widehat{\delta}_{k}\left(  \overline{L}%
(k),\overline{A}\left(  k-1\right)  ,\overline{Z}\left(  k-1\right)  \right)
,$ $k=0,...,J-1$ and let $\widehat{\overline{\mathcal{W}}}^{\dag}$ denote the
corresponding estimated weight. Then under the assumptions given in the lemma
above, we have that $\mathbb{E}\left\{  D_{sm}\left(  h,\beta_{0}\right)
/\overline{\mathcal{W}}^{\#}\right\}  =\mathbb{E}^{\ast}\left\{  D_{sm}\left(
h,\beta_{0}\right)  \right\}  =0$ where $\left\{  D_{sm}\left(  h,\beta
\right)  :h\right\}  $ is the set of unbiased estimating functions of
$\beta_{0}$ corresponding to one of models 1-2 under (ii*). Then, assuming
that $\mathbb{E}\left\{  \nabla_{\beta}D_{sm}\left(  h,\beta\right)
|_{\beta_{0}}/\overline{\mathcal{W}}^{\dag}\right\}  $ is invertible, the
above lemma motivates the following simple weighted estimating equation of the
RAL estimator $\widehat{\beta}_{ipw}:$
\[
o_{p}\left(  n^{-1/2}\right)  =\mathbb{P}_{n}\widehat{D}_{sm}\left(
h,\widehat{\beta}_{ipw}\right)  /\widehat{\overline{\mathcal{W}}}^{\dag}.
\]
The asymptotic distribution of $\widehat{\beta}_{ipw}$ follows from a standard
Taylor expansion and is omitted, the nonparametric bootstrap may also be used
for inference. Note that for estimating models 2.1-2.4 all unknown
expectations must be estimated with a corresponding weighted expectation as
outlined in the previous Section, however now using the modified weights
$\widehat{\overline{\mathcal{W}}}^{\dag}(int(t))$.

\subsection{New multiply robust estimators}

Next, we describe multiply robust estimators of $\beta_{0}$ which is motivated
by considering the set of influence functions associated with RAL estimators
of $\beta_{0}$ in the semiparametric model $\mathcal{M}_{IV\text{ }}$ defined
only by the MSM, the consistency assumption and assumptions (1)-(5).

\begin{lemma}
All RAL estimators $\widehat{\beta}_{np}=\widehat{\beta}_{np}\left(  h\right)
$ of $\beta_{0}$ under $\mathcal{M}_{IV\text{ }}$ are solutions to an
estimating equation of the form
\begin{align*}
o_{p}\left(  n^{-1/2}\right)   &  =\mathbb{P}_{n}\left[  D^{\dag}\left(
h,\widehat{\beta}_{np}\left(  h\right)  \right)  \right] \\
&  =\mathbb{P}_{n}\left[  \frac{D_{sm}\left(  h,\widehat{\beta}_{np}\left(
h\right)  \right)  }{\overline{\mathcal{W}}^{\dag}}\right] \\
&  -\mathbb{P}_{n}\left[  \sum_{j=0}^{J-1}\frac{1}{\overline{\mathcal{W}%
}^{\dag}(j-1)}\left\{  \frac{\left(  -1\right)  ^{1-Z(j)}\Psi_{j}\left(
\widehat{\beta}_{np}\left(  h\right)  \right)  }{f\left(  Z(j)|\overline
{L}(j),\overline{A}(j-1),\overline{Z}(j-1)\right)  }-\widetilde{\Psi}%
_{j}\left(  \widehat{\beta}_{np}\left(  h\right)  \right)  -\frac{\epsilon
_{j}\widetilde{\Psi}_{j}\left(  \widehat{\beta}_{np}\left(  h\right)  \right)
}{\mathcal{W}_{1}^{\dag}(j)}\right\}  \right]
\end{align*}
where%
\[
\Psi_{J-1}\left(  \beta\right)  =\mathbb{E}\left[  \frac{D_{sm}\left(
h,\beta\right)  }{\mathcal{W}_{2}^{\dag}(J-1)\Delta_{J-1}}|\overline{A}\left(
J-2\right)  ,\overline{L}\left(  J-1\right)  ,\overline{Z}\left(  J-1\right)
\right]  ,
\]
for $j=J-2,...,0,$%
\[
\Psi_{j}\left(  \beta\right)  =\mathbb{E}\left[  \frac{\widetilde{\Psi}%
_{j+1}\left(  \beta\right)  }{\mathcal{W}_{2}^{\dag}(j)\Delta_{j}}%
|\overline{A}\left(  j-1\right)  ,\overline{L}\left(  j\right)  ,\overline
{Z}\left(  j\right)  \right]  ,
\]
for $j=J-1,...,0,$
\[
\widetilde{\Psi}_{j}\left(  \beta\right)  =\sum_{z(j)}(-1)^{1-z(j)}\Psi
_{j}\left(  z(j);\beta\right)  ,
\]%
\begin{align*}
\Delta_{j}  &  =\delta_{j}\left(  \overline{A}(j-1),\overline{Z}\left(
j-1\right)  ,\overline{L}\left(  j\right)  \right) \\
\epsilon_{j}  &  =A(j)-\mathbb{E}\left(  A(j)|\overline{A}\left(  j-1\right)
,\overline{L}\left(  j\right)  ,\overline{Z}\left(  j\right)  \right)
\end{align*}

\end{lemma}

The estimator $\widehat{\beta}_{np}\left(  h\right)  $ is not feasible in
practice because it depends on the unknown quantities $\Psi_{j}\left(
\beta\right)  $, $E\left(  A(j)|\overline{A}\left(  j-1\right)  ,\overline
{L}\left(  j\right)  ,\overline{Z}\left(  j\right)  \right)  $ and $f\left(
Z(j)|\overline{L}(j),\overline{A}(j-1),\overline{Z}(j-1)\right)  .$ In
practice, these unknown quantities can be estimated from the observed data
under parametric working models. Let $\Gamma_{j}^{(1)}\left(  \beta\right)
=\widetilde{\Psi}_{j}\left(  \beta\right)  ;\Gamma_{j}^{0}\left(
\beta\right)  =\Psi_{j}\left(  Z\left(  j\right)  =0;\beta\right)  $ with
corresponding estimators $\widehat{\Gamma}_{j}^{(1)}\left(  \beta\right)  $
and $\widehat{\Gamma}_{j}^{(0)}\left(  \beta\right)  $. Likewise, let
\[
\widehat{E}\left(  A(j)|\overline{A}\left(  j-1\right)  ,\overline{L}\left(
j\right)  ,\overline{Z}\left(  j\right)  \right)  =\widehat{\Delta}%
_{j}Z(j)+\widehat{E}\left(  A(j)|\overline{A}\left(  j-1\right)  ,\overline
{L}\left(  j\right)  ,\overline{Z}\left(  j-1\right)  ,Z(j)=0\right)  ,
\]
$\widehat{f}\left(  Z(j)|\overline{L}(j),\overline{A}(j-1),\overline
{Z}(j-1)\right)  $ and $\widehat{\Delta}_{j}$ also denote estimators of
\[
E\left(  A(j)|\overline{A}\left(  j-1\right)  ,\overline{L}\left(  j\right)
,\overline{Z}\left(  j\right)  \right)  =\Delta_{j}Z(j)+E\left(
A(j)|\overline{A}\left(  j-1\right)  ,\overline{L}\left(  j\right)
,\overline{Z}\left(  j-1\right)  ,Z(j)=0\right)  ,
\]
$f\left(  Z(j)|\overline{L}(j),\overline{A}(j-1),\overline{Z}(j-1)\right)  $
and $\Delta_{j}.$ We show in the appendix that the estimator $\widehat{\beta
}_{mr}\left(  h\right)  $ that solves $o_{p}\left(  n^{-1/2}\right)
=\mathbb{P}_{n}\left[  \widehat{D}^{\dag}\left(  h,\widehat{\beta}_{mr}\left(
h\right)  \right)  \right]  $ $\ $where $\widehat{D}^{\dag}$ replaces all
unknown quantities with a corresponding estimator, is multiply robust in the
sense that it is CAN\ if either one but not necessarily all three of the
following conditions hold: (i) $\widehat{\overline{\mathcal{W}}}^{\dag}$ is
consistent for $\overline{\mathcal{W}}^{\dag},$ or (ii) $\widehat{f}\left(
Z(j)|\overline{L}(j),\overline{A}(j-1),\overline{Z}(j-1)\right)  $ is
consistent and $\widehat{\Gamma}_{j}^{1}\left(  \beta\right)  $ is consistent
for all $j\leq J$; or (iii) $\widehat{\Gamma}_{j}^{1}\left(  \beta\right)  $,
$\widehat{\Gamma}_{j}^{0}\left(  \beta\right)  $ and $\widehat{E}\left(
A(j)|\overline{A}\left(  j-1\right)  ,\overline{L}\left(  j\right)
,\overline{Z}\left(  j\right)  \right)  $ are consistent. In Models 2.1-2.3,
the result requires also replacing unknown expectations with corresponding
multiply robust estimators analogous to the estimator given above, details are
omitted. This result effectively generalizes that of Wang and Tchetgen
Tchetgen (2017) to the time-varying setting.

\section{Semiparametric Efficiency}

The semiparametric efficiency bound in a semiparametric model is the inverse
of the variance of the efficient score $S_{eff,\beta_{0}}$ for the model. By
Theorem 5.3 of Newey and McFadden (1993), the efficient score $S_{eff,\beta
_{0}}=$ $D^{\dag}\left(  h_{eff},\beta_{0}\right)  $ in model $\mathcal{M}%
_{IV\text{ }}$ is uniquely characterized by the requirement that for all
$D^{\dag}\left(  h,\beta_{0}\right)  :$%
\begin{equation}
\mathbb{E}\left\{  D^{\dag}\left(  h,\beta_{0}\right)  D^{\dag}\left(
h_{eff},\beta_{0}\right)  ^{T}\right\}  =-\mathbb{E}\left\{  \nabla_{\beta
^{T}}D^{\dag}\left(  h,\beta\right)  |_{\beta_{0}}\right\}  \label{efficient}%
\end{equation}
In order to illustrate the result, consider MSM 1.1$.$ Note that because
$\overline{A}$ is discrete valued with finite support, let $\Xi=\varepsilon
\left(  \beta_{0}\right)  \times\left(  1\left(  \overline{A}=\overline{a}%
_{1}\right)  ,....,1\left(  \overline{A}=\overline{a}_{C}\right)  \right)
^{T}$ where $\left\{  \overline{a}_{c}:c\right\}  $ are the $2^{J}$ possible
values of $\overline{a},$ also let $\mathbf{H=h}\left(  V\right)  $ denote a
$p\times2^{J}$ function of $V.$ The set of influence functions of $\beta_{0}$
under $\mathcal{M}_{IV\text{ }}$ can be written $\left\{  D^{\dag}\left(
\mathbf{h}\right)  :\mathbf{h}\right\}  $ where $D^{\dag}\left(
\mathbf{h}\right)  =\mathbf{H}\widetilde{\Xi}$, and%
\begin{align*}
&  \text{ }\widetilde{\Xi}=\frac{\Xi}{\overline{\mathcal{W}}^{\dag}}%
-\sum_{j=0}^{J-1}\frac{1}{\overline{\mathcal{W}}^{\dag}(j-1)}\left\{
\frac{\left(  -1\right)  ^{1-Z(j)}\mathbb{E}\left[  \frac{\Xi}{\mathcal{W}%
_{2}^{\dag}(J-1)\Delta_{J-1}}|\overline{A}\left(  j-1\right)  ,\overline
{L}\left(  j\right)  ,\overline{Z}\left(  j\right)  \right]  }{f\left(
Z(j)|\overline{L}(j),\overline{A}(j-1),\overline{Z}(j-1)\right)  }\right. \\
&  -\sum_{z(j)}\mathbb{E}\left[  \frac{\Xi}{\mathcal{W}_{2}^{\dag}%
(J-1)\Delta_{J-1}}|\overline{A}\left(  j-1\right)  ,\overline{L}\left(
j\right)  ,\overline{Z}\left(  j-1\right)  ,z(j)\right] \\
&  -\frac{\epsilon_{j}\sum_{z(j)}\mathbb{E}\left[  \frac{\Xi}{\mathcal{W}%
_{2}^{\dag}(J-1)\Delta_{J-1}}|\overline{A}\left(  j-1\right)  ,\overline
{L}\left(  j\right)  ,\overline{Z}\left(  j-1\right)  ,z(j)\right]
}{\mathcal{W}_{1}^{\dag}(j)}%
\end{align*}
A straightforward application of equation $\left(  \ref{efficient}\right)  $
gives the efficient influence function: $D_{eff}^{\dag}\left(  \mathbf{h}%
\right)  =\mathbf{H}_{eff}\widetilde{\Xi}$ where $\mathbf{H}_{eff}%
=\mathbb{E}\left\{  \nabla_{\beta^{T}}\widetilde{\Xi}\left(  \beta\right)
|_{\beta_{0}}|V\right\}  \mathbb{E}\left\{  \widetilde{\Xi}\widetilde{\Xi}%
^{T}|V\right\}  ^{-1}.$ The efficient influence function for other MSMs
considered in this paper can likewise be obtained by straightforward
application of equation $\left(  \ref{efficient}\right)  $ although details
are omitted.

\section{Final Remarks}

This technical report provides identification conditions for MSMs using a
time-varying instrumental variable in the case of time-varying endogenous
binary treatment, a long-standing problem in the causal inference literature.
The case of polytomous or continuous treatments will be discussed elsewhere.
The paper also provides weighted estimating equations that are easy to
implement, as well as multiply robust estimating equations which are
substantially more computationally intensive. Evaluation of final sample
performance and application of these methods is currently underway and will be
published elsewhere.

\begin{center}
{\Large APPENDIX}
\end{center}

\begin{proof}
[Proof of Lemma 1]:
The proof is by induction backwards on the time index $j$. That is, supposing
for some $j, 1\le j\le J-1,$ we have established
\begin{align*}
\mathbb{E}\left(  \frac{g(\overline{A},\overline{L})}{\overline{W}^{\dagger}%
}\right)   &  =\mathbb{E}\left\{  (\overline{W}^{\dagger}_{j})^{-1}%
\sum_{\underline{a_{j+1}}} \mathbb{E}[g(\underline{A}(j+1)=\underline{a_{j+1}%
},\underline{L}_{\underline{A}(j+1)=\underline{a_{j+1}}}(j+2),\overline
{A}(j),\overline{L}(j+1)\mid\overline{LU}(j+1),\overline{AZ}(j)] \right. \\
&  \left.  \quad\times\prod_{k=j+1}^{J-1}f^{\ast}(a_{k}\mid\overline
{A}(j),A_{j+1}=a_{j+1},\ldots,A_{k-1}=a_{k-1})
\vphantom{\sum_{\ul{a_{j+1}}} \EE}\right\}  ,
\end{align*}
we establish the same with $j$ replaced by $j-1$ throughout. In the preceding
display, the notation $A_{j+1}=a_{j+1},\ldots,A_{k-1}=a_{k-1}$ is to be read
as an empty list when $j+1>k-1$; the product $\prod_{k=j+1}^{J-1}(\ldots)$ is
$1$ for $j+1>J-1$; and similarly $L_{a}(j)=L_{a}(J)$ for $j> J$. The summation
$\sum_{\underline{a_{j+1}}}$ ranges over all treatment regimes $a_{j+1}%
,a_{j+2},\ldots,a_{J-1}\in\mathcal{A}^{J-j-1}$; for $j+1>J-1$, the sum
$\sum_{\underline{a_{j+1}}}\mathbb{E}[g(\underline{A}(j+1)=\underline{a_{j+1}%
},\underline{L}_{\underline{A}(j+1)=\underline{a_{j+1}}}(j+2),\overline
{A}(j),\overline{L}(j+1)\mid\overline{LU}(j+1),\overline{AZ}(j)]$ is just
$\mathbb{E}[g(\overline{A}(J-1),\overline{L}(J)\mid\overline{LU}%
(J),\overline{AZ}(J-1)]$. Conditioning on $V$ is assumed throughout, though
suppressed. With these notation conventions, the $j=J-1$ case holds trivially.
Conditioning with respect to $\overline{LU}(j), \overline{AZ}(j-1),$ the rhs
is
\begin{align*}
&  =\mathbb{E}\left\{  (\overline{W}^{\dagger}_{j-1})^{-1}\sum
_{\underline{a_{j+1}}} \mathbb{E}\left\{  (W^{\dagger}_{j})^{-1}%
\mathbb{E}[g(\ldots)\mid\ldots]\prod f^{\ast}(\ldots)\middle| \overline
{LU}(j),\overline{AZ}(j-1)\right\}  \right\}  .
\end{align*}
Considering a single term of the sum,
\begin{align}
&  \mathbb{E}\left\{  \vphantom{\sum_{\ul{dd}}}(W^{\dagger}_{j})^{-1}%
\mathbb{E}[g(\underline{A}(j+1)=\underline{a_{j+1}},\underline{L}%
_{\underline{A}(j+1)=\underline{a_{j+1}}}(j+2),\overline{A}(j),\overline
{L}(j+1)\mid\overline{LU}(j+1),\overline{AZ}(j)] \right. \nonumber\\
&  \left.  \quad\times\prod_{k=j+1}^{J-1}f^{\ast}(a_{k}\mid\overline
{A}(j),A_{j+1}=a_{j+1},\ldots,A_{k-1}=a_{k-1}) \middle| \overline
{LU}(j),\overline{AZ}(j-1)\right\} \nonumber\\
&  =\mathbb{E}\left\{  \vphantom{\sum_{\ul{dd}}}(W^{\dagger}_{j}%
)^{-1}\mathbb{E}[g(\underline{A}(j+1)=\underline{a_{j+1}},\underline{L}%
_{\underline{A}(j+1)=\underline{a_{j+1}}}(j+2),\overline{A}(j),\overline
{L}(j+1)\mid\overline{LUAZ}(j)] \right. \nonumber\\
&  \left.  \quad\times\prod_{k=j+1}^{J-1}f^{\ast}(a_{k}\mid\overline
{A}(j),A_{j+1}=a_{j+1},\ldots,A_{k-1}=a_{k-1}) \middle| \overline
{LU}(j),\overline{AZ}(j-1)\right\} \nonumber\\
&  =\mathbb{E}\left\{  \vphantom{\sum_{\ul{dd}}}(W^{\dagger}_{1,j}%
)^{-1}(-1)^{1-A(j)}\mathbb{E}[g(\underline{A}(j+1)=\underline{a_{j+1}%
},\underline{L}_{\underline{A}(j+1)=\underline{a_{j+1}}}(j+2),\overline
{A}(j),\overline{L}(j+1)\mid\overline{LUAZ}(j)] \right. \nonumber\\
&  \left.  \quad\times f^{\ast}(A(j)\mid\overline{A}(j-1))\prod_{k=j+1}%
^{J-1}f^{\ast}(a_{k}\mid\overline{A}(j),A_{j+1}=a_{j+1},\ldots,A_{k-1}%
=a_{k-1}) \middle| \overline{LU}(j),\overline{AZ}(j-1)\right\} \nonumber\\
&  =\mathbb{E}\left\{  \sum_{a_{j}\in\{0,1\}}\frac{\mathbbm{1}\{A(j)=a_{j}%
\}(-1)^{1-a_{j}}}{W^{\dagger}_{1,j}}\vphantom{\sum_{\ul{dd}}}\mathbb{E}%
[g(\underline{A}(j)=\underline{a_{j}},\underline{L}_{\underline{A}%
(j)=\underline{a_{j}}}(j+1),\overline{A}(j-1),\overline{L}(j)\mid
\overline{LUZ}(j),\overline{A}(j-1),A(j)=a_{j}] \right. \nonumber\\
&  \left.  \quad\times\prod_{k=j}^{J-1}f^{\ast}(a_{k}\mid\overline
{A}(j-1),A_{j}=a_{j},\ldots,A_{k-1}=a_{k-1}) \middle| \overline{LU}%
(j),\overline{AZ}(j-1)\right\}  . \label{lemma:1 eqn:1}%
\end{align}
By SRA, $L_{\underline{A}(j)=\underline{a_{j}}}\amalg A(j) \mid\overline
{A}(j-1),\overline{ZLU}(j),$
\begin{align*}
&  \mathbb{E}[g(\underline{A}(j)=\underline{a_{j}},\underline{L}%
_{\underline{A}(j)=\underline{a_{j}}}(j+1),\overline{A}(j-1),\overline
{L}(j)\mid\overline{LUZ}(j),\overline{A}(j-1),A(j)=a_{j}]\\
&  =\mathbb{E}[g(\underline{A}(j)=\underline{a_{j}},\underline{L}%
_{\underline{A}(j)=\underline{a_{j}}}(j+1),\overline{A}(j-1),\overline
{L}(j)\mid\overline{LUZ}(j),\overline{A}(j-1)],
\end{align*}
and by IV independence, $L_{\underline{A}(j)=\underline{a_{j}}}\amalg
Z(j)\mid\overline{AZ}(j-1),\overline{LU}(j),$
\begin{align*}
&  \mathbb{E}[g(\underline{A}(j)=\underline{a_{j}},\underline{L}%
_{\underline{A}(j)=\underline{a_{j}}}(j+1),\overline{A}(j-1),\overline
{L}(j)\mid\overline{LUZ}(j),\overline{A}(j-1)]\\
&  =\mathbb{E}[g(\underline{A}(j)=\underline{a_{j}},\underline{L}%
_{\underline{A}(j)=\underline{a_{j}}}(j+1),\overline{A}(j-1),\overline
{L}(j)\mid\overline{LU}(j),\overline{AZ}(j-1)].
\end{align*}

so that (\ref{lemma:1 eqn:1}) is
\begin{align}
&  =\sum_{a_{j}\in\{0,1\}}(-1)^{1-a_{j}}\mathbb{E}[g(\underline{A}%
(j)=\underline{a_{j}},\underline{L}_{\underline{A}(j)=\underline{a_{j}}%
}(j+1),\overline{A}(j-1),\overline{L}(j)\mid\overline{LU}(j),\overline
{AZ}(j-1)]\nonumber\\
&  \quad\times\prod_{k=j}^{J-1}f^{\ast}(a_{k}\mid\overline{A}(j-1),A_{j}%
=a_{j},\ldots,A_{k-1}=a_{k-1}) \mathbb{E}\left\{  (W^{\dagger}_{1,j}%
)^{-1}\mathbbm{1}\left\{  A(j)=a_{j}\right\}  \middle| \overline
{LU}(j),\overline{AZ}(j-1)\right\} \nonumber\\
&  =\sum_{a_{j}\in\{0,1\}}(-1)^{1-a_{j}}\mathbb{E}[g(\underline{A}%
(j)=\underline{a_{j}},\underline{L}_{\underline{A}(j)=\underline{a_{j}}%
}(j+1),\overline{A}(j-1),\overline{L}(j)\mid\overline{LU}(j),\overline
{AZ}(j-1)]\nonumber\\
&  \quad\times\frac{\prod_{k=j}^{J-1}f^{\ast}(a_{k}\mid\overline{A}%
(j-1),A_{j}=a_{j},\ldots,A_{k-1}=a_{k-1})}{\delta_{j}(\overline{L}%
(j),\overline{AZ}(j-1))} \mathbb{E}\left\{  \frac{(-1)^{1-Z(j)}%
\mathbbm{1}\left\{  A(j)=a_{j}\right\}  }{f(Z(j)\mid\overline{L}%
(j),\overline{AZ}(j-1))} \middle| \overline{LU}(j),\overline{AZ}(j-1)\right\}
. \label{lemma:1 eqn:2}%
\end{align}
By another application of IV independence,
\begin{align*}
&  \mathbb{E}\left\{  \frac{(-1)^{1-Z(j)}\mathbbm{1}\left\{  A(j)=a_{j}%
\right\}  }{f(Z(j)\mid\overline{L}(j),\overline{AZ}(j-1))} \middle|
\overline{LU}(j),\overline{AZ}(j-1)\right\} \\
&  =\mathbb{E}\left\{  \mathbbm{1}\{A(j)=a_{j}\}\sum_{z_{j}\in\{0,1\}}%
\frac{(-1)^{1-z_{j}}\mathbbm{1}\{Z(j)=z_{j}\}}{f(z_{j}\mid\overline
{L}(j),\overline{AZ}(j-1))} \middle| \overline{LU}(j),\overline{AZ}%
(j-1)\right\} \\
&  =\mathbb{E}\left\{  \mathbb{P}[A(j)=a_{j}\mid\overline{LUZ}(j),\overline
{A}(j-1)]\sum_{z_{j}\in\{0,1\}}\frac{(-1)^{1-z_{j}}\mathbbm{1}\{Z(j)=z_{j}%
\}}{f(z_{j}\mid\overline{L}(j),\overline{AZ}(j-1))} \middle| \overline
{LU}(j),\overline{AZ}(j-1)\right\} \\
&  = \sum_{z_{j}\in\{0,1\}}(-1)^{1-z_{j}}\mathbb{P}[A(j)=a_{j}\mid
\overline{LU}(j),\overline{AZ}(j-1),Z(j)=z_{j}]\frac{\mathbb{P}[Z(j)=z_{j}%
\mid\overline{LU}(j),\overline{AZ}(j-1)]}{f(z_{j}\mid\overline{L}%
(j),\overline{AZ}(j-1))}\\
&  = \sum_{z_{j}\in\{0,1\}}(-1)^{1-z_{j}}\mathbb{P}[A(j)=a_{j}\mid
\overline{LU}(j),\overline{AZ}(j-1),Z(j)=z_{j}]\\
&  = (-1)^{1-a_{j}}\delta_{j}(\overline{L}(j),\overline{AZ}(j-1)).\\
\end{align*}
Therefore, (\ref{lemma:1 eqn:2}) is
\begin{align*}
&  \sum_{a_{j}\in\{0,1\}}\mathbb{E}[g(\underline{A}(j)=\underline{a_{j}%
},\underline{L}_{\underline{A}(j)=\underline{a_{j}}}(j+1),\overline
{A}(j-1),\overline{L}(j)\mid\overline{LU}(j),\overline{AZ}(j-1)]\\
&  \quad\times\prod_{k=j}^{J-1}f^{\ast}(a_{k}\mid\overline{A}(j-1),A_{j}%
=a_{j},\ldots,A_{k-1}=a_{k-1}),
\end{align*}
as required.
\end{proof}

\begin{proof}
[Proof of Lemma 2]Consider the MSM indexed by $\beta\,\ $that solves
\[
\mathbb{E}\left\{  \frac{D_{sm}\left(  h,\beta_{0}\right)  }{\overline
{\mathcal{W}}^{\dag}}\right\}  =0
\]
at $\beta=\beta_{0}=\beta_{0}\left(  h\right)  $ for all $h\in H$, functions
of $\left(  \overline{A},V\right)  $ such that $D_{sm}\left(  h,\beta
_{0}\right)  /\overline{\mathcal{W}}^{\dag}$ is in the Hilbert space $L_{2}$
of functions with finite variance. \ Let $\left\{  F_{t}\left(  \overline
{A},\overline{L}\right)  :t\in\left(  -\epsilon,\epsilon\right)  \right\}  $
denote a regular parametric submodel for a unit's observed data distribution
indexed by a scalar parameter $t$ such that $F_{t=0}\left(  \overline
{A},\overline{L}\right)  =F\left(  \overline{A},\overline{L}\right)  $
generated the observed data. We have that
\[
\mathbb{E}_{t}\left\{  \frac{D_{sm}\left(  h,\beta\left(  F_{t}\right)
\right)  }{\overline{\mathcal{W}}_{t}^{\dag}}\right\}  =0\text{ for all }%
t\in\left(  -\epsilon,\epsilon\right)  ,
\]
and therefore
\begin{align*}
0  &  =\nabla_{t}\mathbb{E}_{t}\left\{  \frac{D_{sm}\left(  h,\beta\left(
F_{t}\right)  \right)  }{\overline{\mathcal{W}}_{t}^{\dag}}\right\} \\
&  =\mathbb{E}\left\{  \frac{D_{sm}\left(  h,\beta\left(  F_{t}\right)
\right)  }{\overline{\mathcal{W}}_{t}^{\dag}}\mathbb{S}\right\}
+\mathbb{E}_{t}\left\{  \frac{D_{sm}\left(  h,\beta\left(  F_{t}\right)
\right)  }{\overline{\mathcal{W}}^{\dag2}}\nabla_{t}\overline{\mathcal{W}}%
_{t}^{\dag}\right\} \\
&  +\mathbb{E}\left\{  \frac{\nabla_{\beta}D_{sm}\left(  h,\beta\right)
}{\overline{\mathcal{W}}_{t}^{\dag}}\right\}  \nabla_{t}\beta\left(
F_{t}\right)
\end{align*}
Consider term
\begin{align*}
&  \mathbb{E}\left\{  \frac{D_{sm}\left(  h,\beta\left(  F\right)  \right)
}{\overline{\mathcal{W}}^{\dag2}}\nabla_{t}\overline{\mathcal{W}}_{t}^{\dag
}\right\} \\
&  =\mathbb{E}\left\{  \sum_{k=0}^{J-1}\frac{D_{sm}\left(  h,\beta\left(
F\right)  \right)  }{\overline{\mathcal{W}}^{\dag2}}\left\{
%TCIMACRO{\dprod \limits_{j\neq k}}%
%BeginExpansion
{\displaystyle\prod\limits_{j\neq k}}
%EndExpansion
\mathcal{W}^{\dag}\left(  j\right)  \right\}  \nabla_{t}\mathcal{W}_{t}^{\dag
}(k)\right\} \\
&  =\mathbb{E}\left\{  \sum_{k=0}^{J-1}\frac{D_{sm}\left(  h,\beta\left(
F\right)  \right)  }{\left\{
%TCIMACRO{\dprod \limits_{j\neq k}}%
%BeginExpansion
{\displaystyle\prod\limits_{j\neq k}}
%EndExpansion
\mathcal{W}^{\dag}\left(  j\right)  \right\}  }\frac{\nabla_{t}\mathcal{W}%
_{t}^{\dag}(k)}{\mathcal{W}^{\dag}(k)^{2}}\right\}
\end{align*}
Next consider term
\begin{align*}
&  \frac{\nabla_{t}\mathcal{W}_{t}^{\dag}(k)}{\mathcal{W}^{\dag}(k)^{2}}\\
&  =\frac{\nabla_{t}\mathcal{W}_{t,1}^{\dag}\left(  k\right)  \mathcal{W}%
_{t,2}^{\dag}(k)}{\mathcal{W}_{1}^{\dag}\left(  k\right)  ^{2}\mathcal{W}%
_{2}^{\dag}(k)^{2}}\\
&  =\frac{\mathbb{S}\left(  Z\left(  k\right)  \right)  }{\mathcal{W}^{\dag
}(k)}\\
&  +\frac{\nabla_{t}\delta_{k,t}\left(  \overline{L}(k),\overline{A}\left(
k-1\right)  ,\overline{Z}\left(  k-1\right)  \right)  }{\mathcal{W}^{\dag
}(k)\delta_{k}\left(  \overline{L}(k),\overline{A}\left(  k-1\right)
,\overline{Z}\left(  k-1\right)  \right)  }%
\end{align*}
where $S\left(  Z\left(  k\right)  \right)  $ is the score function of
$f_{t}\left(  Z(k)|\overline{L}(k),\overline{A}(k-1),\overline{Z}(k-1)\right)
$ and $S\left(  A(k)\right)  $ is the score function of $f_{t}\left(
A(k)|\overline{L}(k),\overline{A}(k-1),\overline{Z}(k)\right)  $. Further
noting that
\begin{align*}
&  \nabla_{t}\delta_{k,t}\left(  \overline{A}\left(  k-1\right)  ,\overline
{Z}\left(  k-1\right)  ,\overline{L}(k)\right) \\
&  =\mathbb{E}\left\{  \mathbb{S}\left(  A(k)\right)  \left(  \frac{\left(
-1\right)  ^{1-Z(k)}\left(
\begin{array}
[c]{c}%
A(k)\\
-\mathbb{E}\left(  A(k)|\overline{A}\left(  k-1\right)  ,\overline{Z}\left(
k-1\right)  ,\overline{L}(k)\right)
\end{array}
\right)  }{f\left(  Z(k)|\overline{L}(k),\overline{A}(k-1),\overline
{Z}(k-1)\right)  }\right)  |\overline{A}\left(  k-1\right)  ,\overline
{Z}\left(  k-1\right)  ,\overline{L}(k)\right\}
\end{align*}
we have%
\begin{align*}
&  \mathbb{E}\left\{  \frac{D_{sm}\left(  h,\beta\left(  F\right)  \right)
}{\overline{\mathcal{W}}^{\dag2}}\nabla_{t}\overline{\mathcal{W}}_{t}^{\dag
}\right\} \\
&  =\mathbb{E}\left\{
\begin{array}
[c]{c}%
\sum_{k=0}^{J-1}\frac{1}{\left\{
%TCIMACRO{\dprod \limits_{j<k}}%
%BeginExpansion
{\displaystyle\prod\limits_{j<k}}
%EndExpansion
\mathcal{W}^{\dag}\left(  j\right)  \right\}  }\frac{\left(  -1\right)
^{1-Z(k)}\mathbb{S}\left(  Z\left(  k\right)  \right)  }{f\left(
Z(k)|\overline{L}(k),\overline{A}(k-1),\overline{Z}(k-1)\right)  }\\
\times\mathbb{E}\left[  \frac{D_{sm}\left(  h,\beta\left(  F\right)  \right)
}{\delta_{k}\left(  \overline{A}\left(  k-1\right)  ,\overline{Z}\left(
k-1\right)  \right)  \mathcal{W}_{2}^{\dag}(k)\left\{
%TCIMACRO{\dprod \limits_{j>k}}%
%BeginExpansion
{\displaystyle\prod\limits_{j>k}}
%EndExpansion
\mathcal{W}^{\dag}\left(  j\right)  \right\}  }|\overline{A}(k-1),\overline
{L}(k),\overline{Z}(k)\right]
\end{array}
\right\} \\
&  +\mathbb{E}\left\{
\begin{array}
[c]{c}%
\sum_{k=0}^{J-1}\frac{1}{\left\{
%TCIMACRO{\dprod \limits_{j<k}}%
%BeginExpansion
{\displaystyle\prod\limits_{j<k}}
%EndExpansion
\mathcal{W}^{\dag}\left(  j\right)  \right\}  }\mathbb{S}\left(  A(k)\right)
\left(  \frac{\left(  -1\right)  ^{1-Z(k)}\left(  A(k)-\mathbb{E}\left(
A(k)|\overline{A}\left(  k-1\right)  ,\overline{Z}\left(  k-1\right)
,\overline{L}(k)\right)  \right)  }{\delta_{k}\left(  \overline{A}\left(
k-1\right)  ,\overline{Z}\left(  k-1\right)  \right)  f\left(  Z(k)|\overline
{L}(k),\overline{A}(k-1),\overline{Z}(k-1)\right)  }\right) \\
\times\mathbb{E}\left[  \frac{D_{sm}\left(  h,\beta\left(  F\right)  \right)
}{\left\{
%TCIMACRO{\dprod \limits_{j\geq k}}%
%BeginExpansion
{\displaystyle\prod\limits_{j\geq k}}
%EndExpansion
\mathcal{W}^{\dag}\left(  j\right)  \right\}  }|\overline{A}(k-1),\overline
{L}(k),\overline{Z}(k)\right]
\end{array}
\right\} \\
&  =\mathbb{E}\left\{
\begin{array}
[c]{c}%
\sum_{k=0}^{J-1}\frac{\mathbb{S}\left(  Z\left(  k\right)  \right)  }{\left\{
%
%TCIMACRO{\dprod \limits_{j<k}}%
%BeginExpansion
{\displaystyle\prod\limits_{j<k}}
%EndExpansion
\mathcal{W}^{\dag}\left(  j\right)  \right\}  }\left(  \frac{\left(
-1\right)  ^{1-Z(k)}}{f\left(  Z(k)|\overline{L}(k),\overline{A}%
(k-1),\overline{Z}(k-1)\right)  }-1\right) \\
\mathbb{\times E}\left[  \frac{D_{sm}\left(  h,\beta\left(  F\right)  \right)
}{\Delta_{k}\mathcal{W}_{2}^{\dag}(k)\left\{
%TCIMACRO{\dprod \limits_{j>k}}%
%BeginExpansion
{\displaystyle\prod\limits_{j>k}}
%EndExpansion
\mathcal{W}^{\dag}\left(  j\right)  \right\}  }|\overline{A}(k-1),\overline
{L}(k),\overline{Z}(k-1)\right]
\end{array}
\right\} \\
&  +\mathbb{E}\left\{
\begin{array}
[c]{c}%
\sum_{k=0}^{J-1}\frac{1}{\left\{
%TCIMACRO{\dprod \limits_{j<k}}%
%BeginExpansion
{\displaystyle\prod\limits_{j<k}}
%EndExpansion
\mathcal{W}^{\dag}\left(  j\right)  \right\}  }\mathbb{S}\left(  A(k)\right)
\left(  \frac{\left(  -1\right)  ^{1-Z(k)}\left(  A(k)-\mathbb{E}\left(
A(k)|\overline{A}\left(  k-1\right)  ,\overline{Z}\left(  k-1\right)
,\overline{L}(k)\right)  \right)  }{\delta_{k}\left(  \overline{A}\left(
k-1\right)  ,\overline{Z}\left(  k-1\right)  \right)  f\left(  Z(k)|\overline
{L}(k),\overline{A}(k-1),\overline{Z}(k-1)\right)  }\right) \\
\times\mathbb{E}\left[  \frac{D_{sm}\left(  h,\beta\left(  F\right)  \right)
}{\left\{
%TCIMACRO{\dprod \limits_{j\geq k}}%
%BeginExpansion
{\displaystyle\prod\limits_{j\geq k}}
%EndExpansion
\mathcal{W}^{\dag}\left(  j\right)  \right\}  }|\overline{A}(k-1),\overline
{L}(k),\overline{Z}(k-1)\right]
\end{array}
\right\} \\
&  =\mathbb{E}\left\{  \mathbb{S}\left(  O\right)  \sum_{k=0}^{J-1}\frac
{1}{\overline{\mathcal{W}}^{\dag}(k-1)}\left\{  \frac{\left(  -1\right)
^{1-Z(k)}\Psi_{k}\left(  \widehat{\beta}_{np}\left(  h\right)  \right)
}{f\left(  Z(k)|\overline{L}(k),\overline{A}(k-1),\overline{Z}(k-1)\right)
}-\widetilde{\Psi}_{k}\left(  \widehat{\beta}_{np}\left(  h\right)  \right)
\right\}  \right\} \\
&  +\mathbb{E}\left\{  \mathbb{S}\left(  O\right)  \sum_{k=0}^{J-1}%
\frac{\epsilon_{k}\widetilde{\Psi}_{k}\left(  \widehat{\beta}_{np}\left(
h\right)  \right)  }{\overline{\mathcal{W}}^{\dag}(k-1)\mathcal{W}_{1}^{\dag
}(k)}\right\}
\end{align*}
Therefire, we conclude that
\begin{align*}
0  &  =\nabla_{t}\mathbb{E}_{t}\left\{  \frac{D_{sm}\left(  h,\beta\left(
F_{t}\right)  \right)  }{\overline{\mathcal{W}}_{t}^{\dag}}\right\} \\
&  =\mathbb{E}\left\{  \frac{D_{sm}\left(  h,\beta\left(  F\right)  \right)
}{\overline{\mathcal{W}}_{t}^{\dag}}\mathbb{S}\right\}  -+\mathbb{E}\left\{
\frac{D_{sm}\left(  h,\beta\left(  F\right)  \right)  }{\overline{\mathcal{W}%
}^{\dag2}}\nabla_{t}\overline{\mathcal{W}}_{t}^{\dag}\right\} \\
&  +\mathbb{E}\left\{  \frac{\nabla_{\beta}D_{sm}\left(  h,\beta\right)
}{\overline{\mathcal{W}}_{t}^{\dag}}\right\}  \nabla_{t}\beta\left(
F_{t}\right) \\
&  =\mathbb{E}\left\{  \mathbb{S}\left(  O\right)  \left[
\begin{array}
[c]{c}%
\frac{D_{sm}\left(  h,\beta\left(  F\right)  \right)  }{\overline{\mathcal{W}%
}^{\dag}}-\sum_{k=0}^{J-1}\frac{1}{\overline{\mathcal{W}}^{\dag}(k-1)}\left\{
\begin{array}
[c]{c}%
\frac{\left(  -1\right)  ^{1-Z(k)}\Psi_{k}\left(  \beta\left(  F\right)
\right)  }{f\left(  Z(k)|\overline{L}(k),\overline{A}(k-1),\overline
{Z}(k-1)\right)  }\\
-\widetilde{\Psi}_{k}\left(  \beta\left(  F\right)  \right)
\end{array}
\right\} \\
-\sum_{k=0}^{J-1}\frac{\epsilon_{k}\widetilde{\Psi}_{k}\left(  \beta\left(
F\right)  \right)  }{\overline{\mathcal{W}}^{\dag}(k-1)\mathcal{W}_{1}^{\dag
}(k)}%
\end{array}
\right]  \right\} \\
&  +\mathbb{E}\left\{  \frac{\nabla_{\beta}D_{sm}\left(  h,\beta\right)
}{\overline{\mathcal{W}}_{t}^{\dag}}\right\}  \nabla_{t}\beta\left(
F_{t}\right)
\end{align*}
proving the result.
\end{proof}

\bigskip

\begin{proof}
[Proof of triple robustness of $\widehat{\beta}_{mr}$]It suffices to show
that
\[
\mathbb{E}\left[
\begin{array}
[c]{c}%
\frac{D_{sm}\left(  h,\beta_{0}\right)  }{\overline{\mathcal{W}}^{\dag\ast}%
}-\sum_{k=0}^{J-1}\frac{1}{\overline{\mathcal{W}}^{\dag\ast}(k-1)}\left\{
\frac{\left(  -1\right)  ^{1-Z(k)}\Psi_{k}^{\ast}\left(  \beta_{0}\right)
}{f^{\ast}\left(  Z(k)|\overline{L}(k),\overline{A}(k-1),\overline
{Z}(k-1)\right)  }-\widetilde{\Psi}_{k}^{\ast}\left(  \beta_{0}\right)
\right\} \\
-\sum_{k=0}^{J-1}\frac{\epsilon_{k}^{\ast}\widetilde{\Psi}_{k}^{\ast}\left(
\beta_{0}\right)  }{\overline{\mathcal{W}}^{\dag\ast}(k-1)\mathcal{W}%
_{1}^{\ast\dag}(k)}%
\end{array}
\right]  =0
\]

provided that either

(i) $\overline{\mathcal{W}}^{\dag\ast}=\overline{\mathcal{W}}^{\dag}$, i.e.
$f^{\ast}\left(  Z(j)|\overline{L}(j),\overline{A}(j-1),\overline
{Z}(j-1)\right)  =f\left(  Z(j)|\overline{L}(j),\overline{A}(j-1),\overline
{Z}(j-1)\right)  $ and \newline$\delta_{j}^{\ast}\left(  \overline
{L}(j),\overline{A}\left(  j-1\right)  ,\overline{Z}\left(  j-1\right)
\right)  =\delta_{j}\left(  \overline{L}(j),\overline{A}\left(  j-1\right)
,\overline{Z}\left(  j-1\right)  \right)  $ for all $j;$ or

(ii) $f^{\ast}\left(  Z(j)|\overline{L}(j),\overline{A}(j-1),\overline
{Z}(j-1)\right)  =f\left(  Z(j)|\overline{L}(j),\overline{A}(j-1),\overline
{Z}(j-1)\right)  $ and $\Gamma_{j}^{1\ast}\left(  \beta_{0}\right)
=\Gamma_{j}^{1}\left(  \beta_{0}\right)  $ for all $j$; or

(iii) $\Gamma_{j}^{1\ast}\left(  \beta_{0}\right)  =\Gamma_{j}^{1}\left(
\beta_{0}\right)  $, $\Gamma_{j}^{0\ast}\left(  \beta_{0}\right)  =\Gamma
_{j}^{0}\left(  \beta_{0}\right)  ,$ and $E^{\ast}\left(  A(j)|\overline
{A}\left(  j-1\right)  ,\overline{L}\left(  j\right)  ,\overline{Z}\left(
j\right)  \right)  =E\left(  A(j)|\overline{A}\left(  j-1\right)
,\overline{L}\left(  j\right)  ,\overline{Z}\left(  j\right)  \right)  .$

The result for (i) holds because $E\left\{  D_{sm}\left(  h,\beta_{0}\right)
/\overline{\mathcal{W}}^{\dag}\right\}  =0,$
\[
\mathbb{E}\left[  \sum_{k=0}^{J-1}\mathbb{E}\left\{
\begin{array}
[c]{c}%
\left(  -1\right)  ^{1-Z(k)}\frac{\Psi_{k}^{\ast}\left(  \beta_{0}\right)
}{f\left(  Z(k)|\overline{L}(k),\overline{A}(k-1),\overline{Z}(k-1)\right)
}\\
-\widetilde{\Psi}_{k}^{\ast}\left(  \beta_{0}\right)  |\overline{L}\left(
k\right)  ,\overline{A}\left(  k-1\right)  ,\overline{Z}\left(  k-1\right)
\end{array}
\right\}  \right]  =0,
\]
and $E\left\{  \epsilon_{k}^{\ast}/\mathcal{W}_{1}^{\dag}(k)|\overline
{A}\left(  k-1\right)  ,\overline{L}\left(  k\right)  ,\overline{Z}\left(
k-1\right)  \right\}  =0,$ because%
\begin{align*}
&  \mathbb{E}\left\{  \frac{\left(  -1\right)  ^{1-Z(k)}\epsilon_{k}^{\ast}%
}{f\left(  Z(k)|\overline{L}(k),\overline{A}(k-1),\overline{Z}(k-1)\right)
\delta_{k}^{\ast}\left(  \overline{L}(k),\overline{A}\left(  k-1\right)
,\overline{Z}\left(  k-1\right)  \right)  }|\overline{A}\left(  k-1\right)
,\overline{L}\left(  k\right)  ,\overline{Z}\left(  k-1\right)  \right\} \\
&  =\mathbb{E}\left\{  \frac{\left(
\begin{array}
[c]{c}%
\left(  -1\right)  ^{1-Z(k)}(A(k)-\delta_{k}\left(  \overline{L}(k),
\overline{A}\left(  k-1\right)  ,\overline{Z}\left(  k-1\right)  \right)
Z(k)\\
-E^{\ast}\left(  A(k)|\overline{A}\left(  k-1\right)  ,\overline{L}\left(
k\right)  ,\overline{Z}\left(  k-1\right)  ,Z(k)=0\right)  )
\end{array}
\right)  }{f\left(  Z(k)|\overline{L}(k),\overline{A}(k-1),\overline
{Z}(k-1)\right)  \delta_{k}\left(  \overline{L}(k),\overline{A}\left(
k-1\right)  ,\overline{Z}\left(  k-1\right)  \right)  }|\overline{A}\left(
k-1\right)  ,\overline{L}\left(  k\right)  ,\overline{Z}\left(  k-1\right)
\right\} \\
&  =\frac{-\delta_{k}\left(  \overline{L}(k),\overline{A}\left(  k-1\right)
,\overline{Z}\left(  k-1\right)  \right)  -\delta_{k}\left(  \overline
{L}(k),\overline{A}\left(  k-1\right)  ,\overline{Z}\left(  k-1\right)
\right)  }{\delta_{k}\left(  \overline{L}(k),\overline{A}\left(  k-1\right)
,\overline{Z}\left(  k-1\right)  \right)  }\\
&  +\sum_{z(k)}\left(  -1\right)  ^{1-z(k)}\frac{
\begin{array}
[c]{c}%
-E^{\ast}\left(  A(k)|\overline{A}\left(  k-1\right)  ,\overline{L}\left(
k\right)  ,\overline{Z}\left(  k-1\right)  ,Z(k)=0\right)
\end{array}
}{\delta_{k}\left(  \overline{L}(k),\overline{A}\left(  k-1\right)
,\overline{Z}\left(  k-1\right)  \right)  }\\
&  =0
\end{align*}
Next, suppose that

(ii) $f^{\ast}\left(  Z(j)|\overline{L}(j),\overline{A}(j-1),\overline
{Z}(j-1)\right)  =f\left(  Z(j)|\overline{L}(j),\overline{A}(j-1),\overline
{Z}(j-1)\right)  $ and $\Gamma_{j}^{(1)\ast}\left(  \beta_{0}\right)
=\Gamma_{j}^{(1)}\left(  \beta_{0}\right)  =\widetilde{\Psi}_{j}\left(
\beta_{0}\right)  $ for all $j$. Then
\begin{align*}
&  \mathbb{E}\left[  \sum_{k=0}^{J-1}\frac{1}{\overline{\mathcal{W}}^{\dag
\ast}(k-1)}\left\{  \frac{\left(  -1\right)  ^{1-Z(k)}\Psi_{k}^{\ast}\left(
\beta_{0}\right)  }{f\left(  Z(k)|\overline{L}(k),\overline{A}(k-1),\overline
{Z}(k-1)\right)  }-\widetilde{\Psi}_{k}\left(  \beta_{0}\right)  \right\}
\right] \\
&  =\mathbb{E}\left[  \sum_{k=0}^{J-1}\frac{1}{\overline{\mathcal{W}}%
^{\dag\ast}(k-1)}\mathbb{E}\left\{
\begin{array}
[c]{c}%
\frac{\left(  -1\right)  ^{1-Z(k)}\Psi_{k}^{\ast}\left(  \beta_{0}\right)
}{f\left(  Z(k)|\overline{L}(k),\overline{A}(k-1),\overline{Z}(k-1)\right)
}\\
-\widetilde{\Psi}_{k}\left(  \beta_{0}\right)  |\overline{L}(k),\overline
{A}(k-1),\overline{Z}(k-1)
\end{array}
\right\}  \right] \\
&  =0
\end{align*}

furthermore,%
\begin{align*}
&  \mathbb{E}\left[  \frac{\epsilon_{J-1}^{\ast}\widetilde{\Psi}_{J-1}\left(
\beta_{0}\right)  }{\overline{\mathcal{W}}^{\dag\ast}(J-2)\mathcal{W}%
_{1}^{\dag\ast}(J-1)}\right] \\
&  =\mathbb{E}\left[  \left(  -1\right)  ^{1-Z\left(  J-1\right)  }%
\frac{\left(
\begin{array}
[c]{c}%
(A(J-1)-\delta_{J-1}^{\ast}\left(  \overline{L}(J-1),\overline{A}\left(
J-2\right)  ,\overline{Z}\left(  J-2\right)  \right)  Z(J-1)\\
-E^{\ast}\left(  A(J-1)|\overline{A}\left(  J-2\right)  ,\overline{L}\left(
J-1\right)  ,\overline{Z}\left(  J-2\right)  ,Z(J-1)=0\right)  )
\widetilde{\Psi}_{J-1}\left(  \beta_{0}\right)
\end{array}
\right)  }{\overline{\mathcal{W}}^{\dag\ast}(J-2)f\left(  Z(J-1)|\overline
{L}(J-1),\overline{A}(J-2),\overline{Z}(J-2)\right)  \delta_{J-1}^{\ast
}\left(  \overline{L}(J-1),\overline{A}\left(  J-2\right)  ,\overline
{Z}\left(  J-2\right)  \right)  }\right] \\
&  =\mathbb{E}\left[  \left(  -1\right)  ^{1-Z\left(  J-1\right)  }%
\frac{\left(
\begin{array}
[c]{c}%
(A(J-1)-\delta_{J-1}^{\ast}\left(  \overline{L}(J-1),\overline{A}\left(
J-2\right)  ,\overline{Z}\left(  J-2\right)  \right)  Z(J-1)\\
-E^{\ast}\left(  A(J-1)|\overline{A}\left(  J-2\right)  ,\overline{L}\left(
J-1\right)  ,\overline{Z}\left(  J-2\right)  ,Z(J-1)=0\right)  )
\widetilde{\Psi}_{J-1}\left(  \beta_{0}\right)
\end{array}
\right)  }{\overline{\mathcal{W}}^{\dag\ast}(J-2)f\left(  Z(J-1)|\overline
{L}(J-1),\overline{A}(J-2),\overline{Z}(J-2)\right)  \delta_{J-1}^{\ast
}\left(  \overline{L}(J-1),\overline{A}\left(  J-2\right)  ,\overline
{Z}\left(  J-2\right)  \right)  }\right] \\
&  =\mathbb{E}\left[  \frac{\delta_{J-1}\left(  \overline{L}(J-1),\overline
{A}\left(  J-2\right)  ,\overline{Z}\left(  J-2\right)  \right)  \Gamma
_{J-1}^{(1)}\left(  \beta_{0}\right)  }{\overline{\mathcal{W}}^{\dag\ast
}(J-2)\delta_{J-1}^{\ast}\left(  \overline{L}(J-1),\overline{A}\left(
J-2\right)  ,\overline{Z}\left(  J-2\right)  \right)  }\right]  -\mathbb{E}%
\left[  \frac{\delta_{J-1}^{\ast}\left(  \overline{L}(J-1),\overline{A}\left(
J-2\right)  ,\overline{Z}\left(  J-2\right)  \right)  \Gamma_{J-1}%
^{(1)}\left(  \beta_{0}\right)  }{\overline{\mathcal{W}}^{\dag\ast}%
(J-2)\delta_{J-1}^{\ast}\left(  \overline{L}(J-1),\overline{A}\left(
J-2\right)  ,\overline{Z}\left(  J-2\right)  \right)  }\right] \\
&  =\mathbb{E}\left[  \frac{\Delta_{J-1}\Gamma_{J-1}^{(1)}\left(  \beta
_{0}\right)  }{\overline{\mathcal{W}}^{\dag\ast}(J-2)\Delta_{J-1}^{\ast}%
}\right]  -\mathbb{E}\left[  \frac{\Gamma_{J-1}^{(1)}\left(  \beta_{0}\right)
}{\overline{\mathcal{W}}^{\dag\ast}(J-2)}\right]
\end{align*}
Likewise for all $1\le j<J-1$%
\[
\mathbb{E}\left[  \frac{\epsilon_{j}^{\ast}\widetilde{\Psi}_{j}\left(
\beta_{0}\right)  }{\overline{\mathcal{W}}^{\dag\ast}(j-1)\mathcal{W}%
_{1}^{\dag\ast}(j)}\right]  =\mathbb{E}\left[  \frac{\Delta_{j}\Gamma
_{j}^{(1)}\left(  \beta_{0}\right)  }{\overline{\mathcal{W}}^{\dag\ast
}(j-1)\Delta_{j}^{\ast}}\right]  -\mathbb{E}\left[  \frac{\Gamma_{j}%
^{(1)}\left(  \beta_{0}\right)  }{\overline{\mathcal{W}}^{\dag\ast}%
(j-1)}\right]
\]
and%
\begin{align*}
&  \mathbb{E}\left[  \frac{\widetilde{\Psi}_{j}\left(  \beta_{0}\right)
}{\overline{\mathcal{W}}^{\dag\ast}(j-1)}\right] \\
&  =\mathbb{E}\left[
\begin{array}
[c]{c}%
\frac{1}{\overline{\mathcal{W}}^{\dag\ast}(j-2)}\frac{\left(  -1\right)
^{1-Z(j-1)}}{f\left(  Z(j-1)|\overline{L}(j-1),\overline{A}(j-2),\overline
{Z}(j-2)\right)  }\\
\times\mathbb{E}\left[  \frac{\widetilde{\Psi}_{j}\left(  \beta_{0}\right)
}{\mathcal{W}_{2}^{\dag\ast}(j-1)\Delta_{J-1}^{\ast}}|\overline{L}%
(j-1),\overline{A}(j-2),\overline{Z}(j-1)\right]
\end{array}
\right] \\
&  =\mathbb{E}\left[  \frac{1}{\overline{\mathcal{W}}^{\dag\ast}(j-2)}%
\frac{\left(  -1\right)  ^{1-Z(j-1)}}{f\left(  Z(j-1)|\overline{L}%
(j-1),\overline{A}(j-2),\overline{Z}(j-2)\right)  }\frac{\Psi_{j-1}^{\ast
}\left(  \beta\right)  \Delta_{j-1}}{\Delta_{j-1}^{\ast}}\right] \\
&  =\mathbb{E}\left[  \frac{\widetilde{\Psi}_{j-1}\Delta_{j-1}}{\overline
{W}^{\dag\ast}(j-2)\Delta^{*}_{j-1}}\right]
\end{align*}
Therefore%
\begin{align*}
&  \mathbb{E}\left[
\begin{array}
[c]{c}%
\frac{D_{sm}\left(  h,\beta_{0}\right)  }{\overline{\mathcal{W}}^{\dag\ast}%
}-\sum_{k=0}^{J-1}\frac{1}{\overline{\mathcal{W}}^{\dag\ast}(k-1)}\left\{
\frac{\left(  -1\right)  ^{1-Z(k)}\Psi_{k}^{\ast}\left(  \beta_{0}\right)
}{f^{\ast}\left(  Z(k)|\overline{L}(k),\overline{A}(k-1),\overline
{Z}(k-1)\right)  }-\widetilde{\Psi}_{k}\left(  \beta_{0}\right)  \right\} \\
-\sum_{k=0}^{J-1}\frac{\epsilon_{k}^{\ast}\widetilde{\Psi}_{k}\left(
\beta_{0}\right)  }{\overline{\mathcal{W}}^{\dag\ast}(k-1)\mathcal{W}%
_{1}^{\dag\ast}(k)}%
\end{array}
\right] \\
&  =\mathbb{E}\left[  \frac{D_{sm}\left(  h,\beta_{0}\right)  }{\overline
{\mathcal{W}}^{\dag\ast}}\right]  -\mathbb{E}\left[  \sum_{k=0}^{J-1}%
\frac{\epsilon_{k}^{\ast}\widetilde{\Psi}_{k}\left(  \beta_{0}\right)
}{\overline{\mathcal{W}}^{\dag\ast}(k-1)\mathcal{W}_{1}^{\dag\ast}(k)}\right]
\\
&  =\mathbb{E}\left[  \frac{D_{sm}\left(  h,\beta_{0}\right)  }{\overline
{\mathcal{W}}^{\dag\ast}}\right]  - \sum_{k=0}^{J-1}\left\{  \mathbb{E}\left[
\frac{\Delta_{k}\Gamma_{k}^{(1)}\left(  \beta_{0}\right)  }{\overline
{\mathcal{W}}^{\dag\ast}(k-1)\Delta_{k}^{\ast}}\right]  -\mathbb{E}\left[
\frac{\Gamma_{k}^{(1)}\left(  \beta_{0}\right)  }{\overline{\mathcal{W}}%
^{\dag\ast}(k-1)}\right]  \right\} \\
&  =\mathbb{E}\left[  \frac{\Delta_{J-1}\Gamma_{J-1}^{(1)}\left(  \beta
_{0}\right)  }{\overline{\mathcal{W}}^{\dag\ast}(J-2)\Delta_{J-1}^{\ast}%
}\right]  - \sum_{k=0}^{J-1}\left\{  \mathbb{E}\left[  \frac{\Delta_{k}%
\Gamma_{k}^{(1)}\left(  \beta_{0}\right)  }{\overline{\mathcal{W}}^{\dag\ast
}(k-1)\Delta_{k}^{\ast}}\right]  -\mathbb{E}\left[  \frac{\Gamma_{k}%
^{(1)}\left(  \beta_{0}\right)  }{\overline{\mathcal{W}}^{\dag\ast}%
(k-1)}\right]  \right\} \\
&  =\mathbb{E}\left[  \frac{\Delta_{J-1}\Gamma_{J-1}^{(1)}\left(  \beta
_{0}\right)  }{\overline{\mathcal{W}}^{\dag\ast}(J-2)\Delta_{J-1}^{\ast}%
}\right]  -\mathbb{E}\left[  \frac{\Delta_{J-1}\Gamma_{J-1}^{(1)}\left(
\beta_{0}\right)  }{\overline{\mathcal{W}}^{\dag\ast}(J-2)\Delta_{J-1}^{\ast}%
}\right]  +\sum_{k=0}^{J-1}\left\{
\begin{array}
[c]{c}%
\mathbb{E}\left[  \frac{\Gamma_{k}^{(1)}\left(  \beta_{0}\right)  }%
{\overline{\mathcal{W}}^{\dag\ast}(k-1)}\right] \\
-\mathbb{E}\left[  \frac{\Delta_{k-1}\Gamma_{k-1}^{(1)}\left(  \beta
_{0}\right)  }{\overline{\mathcal{W}}^{\dag\ast}(k-2)\Delta_{k-1}^{\ast}%
}\right]
\end{array}
\right\} \\
&  =\mathbb{E}\left[  \Gamma_{0}^{(1)}\left(  \beta_{0}\right)  \right]  =0
\end{align*}

where $\Delta_{-1}\equiv0$. $\overline{\mathcal{W}}^{\dag\ast}(-2)=\overline
{\mathcal{W}}^{\dag\ast}(-1)=1,\Delta_{-1}^{\ast}=1.$ Finally, \ suppose that
(iii) $\Gamma_{j}^{1\ast}\left(  \beta_{0}\right)  =\Gamma_{j}^{1}\left(
\beta_{0}\right)  $, $\Gamma_{j}^{0\ast}\left(  \beta_{0}\right)  =\Gamma
_{j}^{0}\left(  \beta_{0}\right)  $ and $E^{\ast}\left(  A(j)|\overline
{A}\left(  j-1\right)  ,\overline{L}\left(  j\right)  ,\overline{Z}\left(
j\right)  \right)  =E\left(  A(j)|\overline{A}\left(  j-1\right)
,\overline{L}\left(  j\right)  ,\overline{Z}\left(  j\right)  \right)  ,$ then
note that
\begin{align*}
&  \mathbb{E}\left\{  \sum_{k=0}^{J-1}\frac{\epsilon_{k}^{\ast}\widetilde{\Psi
}_{k}^{\ast}\left(  \beta_{0}\right)  }{\overline{\mathcal{W}}^{\dag\ast
}(k-1)\mathcal{W}_{1}^{\dag\ast}(k)}\right\} \\
&  =\mathbb{E}\left\{  \sum_{k=0}^{J-1}\frac{\mathbb{E}\left[  \epsilon
_{k}^{\ast}|\overline{A}\left(  k-1\right)  ,\overline{L}\left(  k\right)
,\overline{Z}\left(  k\right)  \right]  \widetilde{\Psi}_{k}^{\ast}\left(
\beta_{0}\right)  }{\overline{\mathcal{W}}^{\dag\ast}(k-1)\mathcal{W}%
_{1}^{\dag\ast}(k)}\right\} \\
&  =\mathbb{E}\left\{  \sum_{k=0}^{J-1}\frac{\mathbb{E}\left[  \epsilon
_{k}|\overline{A}\left(  k-1\right)  ,\overline{L}\left(  k\right)
,\overline{Z}\left(  k\right)  \right]  \widetilde{\Psi}_{k}^{\ast}\left(
\beta_{0}\right)  }{\overline{\mathcal{W}}^{\dag\ast}(k-1)\mathcal{W}%
_{1}^{\dag\ast}(k)}\right\} \\
&  =0
\end{align*}
Therefore
\begin{align*}
&  \mathbb{E}\left[
\begin{array}
[c]{c}%
\frac{D_{sm}\left(  h,\beta_{0}\right)  }{\overline{\mathcal{W}}^{\dag\ast}%
}-\sum_{k=0}^{J-1}\frac{1}{\overline{\mathcal{W}}^{\dag\ast}(k-1)}\left\{
\frac{\left(  -1\right)  ^{1-Z(k)}\Psi_{k}^{\ast}\left(  \beta_{0}\right)
}{f^{\ast}\left(  Z(k)|\overline{L}(k),\overline{A}(k-1),\overline
{Z}(k-1)\right)  }-\widetilde{\Psi}_{k}\left(  \beta_{0}\right)  \right\} \\
-\sum_{k=0}^{J-1}\frac{\epsilon_{k}^{\ast}\widetilde{\Psi}_{k}\left(
\beta_{0}\right)  }{\overline{\mathcal{W}}^{\dag\ast}(k-1)\mathcal{W}%
_{1}^{\dag\ast}(k)}%
\end{array}
\right] \\
&  =\mathbb{E}\left[  \frac{D_{sm}\left(  h,\beta_{0}\right)  }{\overline
{\mathcal{W}}^{\dag\ast}}-\sum_{k=0}^{J-1}\frac{1}{\overline{\mathcal{W}%
}^{\dag\ast}(k-1)}\left\{  \frac{\left(  -1\right)  ^{1-Z(k)}\Psi_{k}\left(
\beta_{0}\right)  }{f^{\ast}\left(  Z(k)|\overline{L}(k),\overline
{A}(k-1),\overline{Z}(k-1)\right)  }-\widetilde{\Psi}_{k}\left(  \beta
_{0}\right)  \right\}  \right] \\
&  =\mathbb{E}\left[  \frac{D_{sm}\left(  h,\beta_{0}\right)  -\Psi
_{J-1}\left(  \beta_{0}\right)  \Delta_{J-1}\mathcal{W}_{2}^{\dag\ast}%
(J-1)}{\overline{\mathcal{W}}^{\dag\ast}}+\sum_{k=0}^{J-1}\frac
{\widetilde{\Psi}_{k}\left(  \beta_{0}\right)  -\Psi_{k-1}\left(  \beta
_{0}\right)  \Delta_{k-1}\mathcal{W}_{2}^{\dag\ast}(k-1)}{\overline
{\mathcal{W}}^{\dag\ast}(k-1)}\right] \\
&  =\mathbb{E}\left[  \Gamma_{0}^{(1)}\left(  \beta_{0}\right)  \right]  =0,
\end{align*}
proving the result. \newline
\end{proof}

\end{document}